\documentclass[11pt]{article}
\usepackage{fullpage}
\usepackage{subfig}
\usepackage[usenames,dvipsnames]{xcolor}
\usepackage[colorlinks,citecolor=blue,linkcolor=BrickRed]{hyperref}
\usepackage{times}
	\usepackage{latexsym,graphicx,epsfig,color}
\usepackage{amsfonts,amssymb,amsmath,amsthm,amstext}
\usepackage{enumitem}
\usepackage{url,setspace}
\usepackage{multirow}
\usepackage{rotating}
\usepackage{makeidx}
\usepackage{tikz}
\usepackage{enumitem}
\usepackage{accents}
\usepackage{xspace}
\usepackage{algorithm,algpseudocode}
\usepackage{bm}

\newcommand{\IGNORE}[1]{}

\usetikzlibrary{decorations.markings}
\usetikzlibrary{arrows}
\tikzstyle{block}=[draw opacity=0.7,line width=1.4cm]
\tikzstyle{graphnode}=[circle, draw, fill=black!20, inner sep=0pt, minimum width=6pt]
\tikzstyle{point}=[circle, draw, fill=black!30, inner sep=0pt, minimum width=1pt]
\tikzstyle{input}=[rectangle, draw, fill=black!75,inner sep=3pt, inner ysep=3pt, minimum width=4pt]
\tikzstyle{unmatched}=[graphnode,fill=black!0]
\tikzstyle{shaded}=[graphnode,fill=black!20]
\tikzstyle{matched}=[graphnode,fill=black!100]  	
\tikzstyle{matching} = [ultra thick]
\tikzset{
    >=stealth',
    pil/.style={
           ->,
           thick,
           shorten <=2pt,
           shorten >=2pt,}
}
\tikzset{->-/.style={decoration={
  markings,
  mark=at position .5 with {\arrow{>}}},postaction={decorate}}}
\makeindex

\DeclareMathOperator{\argmax}{arg max}

\newtheorem{theorem}{Theorem}[section]
\newtheorem{claim}[theorem]{Claim}

\newtheorem{lemma}[theorem]{Lemma}

\newtheorem{observation}[theorem]{Observation}
\newtheorem{assumption}[theorem]{Assumption}

\newtheorem{defn}[theorem]{Definition}

\usepackage{thmtools,thm-restate} % Used to restate lemmas without new introducing new numbers. 
\newcommand{\E}{\mathbb{E}}

\renewcommand{\S}{\ensuremath{\mathcal{S}}\xspace}

\newcommand{\F}{\mathcal{F}}
\newcommand{\M}{\mathcal{M}}
\newcommand{\T}{\mathcal{T}}

\def\C   {\mathcal{C}}

\newif\ifFULL
\FULLtrue

\newcounter{note}[section]

\newcommand{\price}{\pi}
\newcommand{\probed}{\mathsf{Probed}}
\newcommand{\val}{\mathsf{val}}
\newcommand{\cost}{\mathsf{cost}}
\newcommand{\I}{\mathbb{I}}

\newcommand{\X}{\mathbf{X}}
\newcommand{\Y}{\mathbf{Y}}
\newcommand{\cc}{\mathbf{c}}
\newcommand{\R}{\mathbb{R}_{\geq 0}}
\newcommand{\J}{\mathcal{J}}
\newcommand{\A}{\ensuremath{\mathcal{A}}\xspace}

\newcommand{\frugal}{\text{\sc{Frugal}}\xspace}
\newcommand{\POI}{PoI\xspace}
\newcommand{\FreeInfo}{Free-Info\xspace}
\newcommand{\taumax}{\tau^{\max}}
\newcommand{\taumin}{\tau^{\min}}
\newcommand{\Ym}{\Y_{M}}
\newcommand{\Ymax}{Y^{\max}}
\newcommand{\bYmax}{\mathbf{Y^{\max}}}
\newcommand{\Ymmax}{\Ym^{\max}}
\newcommand{\bYmin}{\mathbf{Y^{\min}}}
\newcommand{\Ymin}{Y^{\min}}
\newcommand{\Clients}{\text{\sc{Clients}}}
\newcommand{\leaf}{l}
\newcommand{\one}{\textsf{1}\xspace}

\title{ %\vspace{-2cm} 
The Price of Information in\\
Combinatorial Optimization\footnote{Part of this work was done while the authors were visiting the Simons Institute for the Theory of Computing.
We are grateful to Anupam Gupta and Viswanath Nagarajan for several discussions on this project.}}
	\author{
Sahil Singla
\thanks{Computer Science Department, Carnegie Mellon
     University, Pittsburgh, PA 15213, USA. Email: \texttt{ssingla@cmu.edu}. Research  supported in part by a CMU Presidential Fellowship and
     NSF awards CCF-1319811, CCF-1536002, CCF-1540541, and CCF-1617790.}
}

\date{\today }

\begin{document}
\maketitle
\thispagestyle{empty}

\begin{abstract}{

Consider a network design  application where we wish to lay down a minimum-cost spanning tree in a given graph; however, we only have  stochastic information about the edge costs. To learn the precise   cost of any edge, we have to conduct a study that incurs a price. Our goal is to find a spanning tree while minimizing the \emph{disutility}, which is the sum of the tree cost  and the total price that we spend on the studies. 
In a different application, each edge gives a stochastic \emph{reward value}. Our goal is to find a spanning tree while  maximizing the \emph{utility}, which is the tree reward \emph{minus} the prices that we pay. 

Situations such as the above two often arise in practice where we wish to find a good solution to an optimization problem, but we  start with only some partial knowledge about the parameters of the problem.  The missing information can be found only  after paying a \emph{probing} price, which we call the  \emph{price of information}. What strategy should we adopt to optimize our expected utility/disutility?  

A classical example of the above setting is Weitzman's ``Pandora's box'' problem where we are given probability distributions on values of $n$ independent random variables. The goal is to choose a \emph{single} variable with a large value,  but we can find the actual outcomes only after  paying a  price. Our work is a generalization of this model to other combinatorial optimization problems such as matching, set cover, facility location, and prize-collecting Steiner tree. We give a technique that reduces such problems to their non-price counterparts, and use it to design exact/approximation algorithms  to optimize our utility/disutility. Our techniques  extend to situations where there are additional constraints on what parameters can be probed or when we can simultaneously probe a subset of the parameters.

%\snote{Add discussion section with open problems.}

}\end{abstract}

%\tableofcontents

%%%%%%%%%%%%%%%%%%%%%%%%%%%%%%
%\newpage
\clearpage
\setcounter{page}{1}

\section{Introduction}\label{sec:intro}
Suppose we want to purchase a house. We have some idea about the value of every available house in the market, say based on its location, size, and photographs. However, to find the exact value of a house we have to hire a house inspector and pay her a price.
Our \emph{utility} is the difference in the value of the best house that we find and the total inspection prices that we pay. We want to design a strategy to maximize our  utility.

The above problem can be modeled as Weitzman's ``Pandora's box'' problem~\cite{Weitzman-Econ79}. Given probability distributions of $n$ independent random variables $X_i$ and given their probing prices $\pi_i$, 
 the  problem is to adaptively {probe} a subset $\probed \subseteq [n]$  to maximize the expected utility: 
 \[ \E \left[ \max_{i\in \probed} \{X_i\} - \sum_{i\in \probed} \pi_i \right].  \]
 Weitzman gave an optimal adaptive strategy that maximizes the expected utility\ifFULL ~(na\"{\i}ve greedy algorithms can behave arbitrarily bad: see  Section~\ref{sec:egNaiveGreedy}). \else. \fi
 However, suppose instead of probing values of elements, we probe weights of edges in a graph. Our utility is the maximum-weight matching that we find \emph{minus}  the total probing prices that we pay. What strategy should we adopt to maximize our expected utility?

In a different scenario, consider a network design \emph{minimization} problem. Suppose we wish to lay down a minimum-cost spanning tree in a given graph; however, we only have  stochastic information about the edge costs. To find the precise cost $X_i$ of any edge, we have to conduct a study that incurs a price $\price_i$. Our \emph{disutility} is the sum of the tree cost  and the total price that we spend on the studies. We want to design a strategy to minimize our expected disutility. An important difference between these two scenarios is that of maximizing utility vs minimizing disutility. %This plays a crucial role in designing  approximation algorithms.

%Our goal is to find a spanning tree while minimizing the total \emph{disutility}, which is the sum of the tree cost  and the price that we spend on studies.

%Consider another example where we want to find a minimum cost set cover for universe $V$ of elements\footnote{Given $n$ sets  $S_1, S_2, \ldots, S_n \subseteq V$ where $i$th set has cost $c_i$, the  \emph{minimum cost set cover} problem is to select some sets such that their union is $V$ and the sum of the selected sets costs is minimized.}.  Suppose the $n$  sets  $S_1, S_2, \ldots, S_n \subseteq V$  are known, however, cost  $X_i$ of the $i$th set is given by a stochastic independent random variable. We can  learn  $X_i$ but only by paying a {probing} price $\price_i$. The problem is to find a strategy to adaptively probe sets to minimize our expected disutility, which is the sum of the costs of our set cover solution and the probing prices. Note that a crucial difference between this and the house purchasing problem is that this is a minimization problem.

Situations like the above often arise where we wish to find a ``good" solution to an optimization problem; however, we  start with only some partial knowledge about the parameters of the problem. 
 The missing information can be found only  after  paying a \emph{probing} price, which we call the  \emph{price of information}. What strategy should we adopt to optimize our expected utility/disutility?
In this work we design optimal/approximation algorithms for several  combinatorial optimization problems in an uncertain environment where we jointly optimize the value of the solution and the price of information.

%The house purchasing problem is an example of the classical Weitzman's ``Pandora's box'' problem~\cite{Weitzman-Econ79}. Given independent value distributions $X_i$ and deterministic prices $\price_i$ of $n$ elements,  the  problem is to adaptively {probe} a subset $\probed$ of them  to find a large value element to maximize the expected \emph{utility}, which  is the difference of value $\max_{i \in \probed} \{X_i \}$ and probing prices $\sum_{i\in \probed} \price_i$, i.e. \[ \max_{i \in \probed} \{X_i \} -  \sum_{i\in \probed} \price_i. \]

%--------------------------------------------------
\subsection{Utility/Disutility Optimization %: Model and Results
}\label{sec:modelResults}
To begin, the above maximum-weight matching problem can be formally modeled as follows.

\noindent \textbf{Max-Weight Matching}  ~ Given a graph $G$ with edges $E$, suppose each edge $i\in E$ takes some random  weight $X_i$ independently from a known probability distribution. 
We can find the exact outcome $X_i$ only after paying a \emph{probing price} $\price_i$. The goal is to adaptively probe a set of edges $\probed \subseteq E$ and select a matching $\I \subseteq \probed$  to  maximize the   expected \emph{utility}, 
 \[ \E \left[ \sum_{i\in \I} X_i - \sum_{i\in \probed} \price_i \right],
 \]
where the expectation is over random variables $\X = (X_1, \ldots, X_n)$ and any internal randomness of the algorithm. We observe that we can only select an edge if it has been probed and we might select only a subset of the  probed edges.  
This matching  problem can be used to model kidney exchanges where testing compatibility of donor-receiver pairs has an associated price.

To capture value functions of more general  combinatorial problems in a single  framework, we define the  notion of \emph{semiadditive} functions.
\begin{defn}[Semiadditive function] 
We say a function $f(\I,\X): 2^{V}\times \R^{|V|} \rightarrow \R$ is \emph{semiadditive} if there exists a function $ h:2^V \rightarrow \R$  such that
\[	f(\I,\X) = \sum_{i\in \I} X_i + h(\I).
\]
\end{defn}
For example, in the case of max-weight matching our value function $f(\I,\X) = \sum_{i\in \I} X_i$ is \emph{additive}, i.e. $h(\I)=0$. We call these functions semiadditive because the second term $ h(\I)$ is allowed to effect the function in a ``non-additive" way; however, not depending on $\X$. Here are some other examples. 
\begin{itemize}
%\item Additive: When $f(\I,X) = \sum_{i\in \I} X_i$.
\item \emph{Uncapacitated  Facility Location}: Given a graph $G=(V,E)$ with metric $(V,d)$, $\Clients \subseteq V$, and facility opening costs $\X: V\rightarrow \R$, we wish to open facilities at some locations $\I \subseteq V$. The function is the sum of facility opening costs and the connection costs to \Clients. Hence,
\begin{align}  \label{eq:facilLocation}
f(\I,\X) = \sum_{i\in \I} X_i + \sum_{j\in \Clients} \min_{i\in \I }d(j,i).
\end{align}
Here $h(\I) =\sum_{j\in \Clients} \min_{i\in \I }d(j,i)$ only depends on $\I$, and not on facility opening costs $\X$.

\item \emph{Prize-Collecting Steiner Tree}: Given a graph $G=(V,E)$ with some edge costs $\cc :E \rightarrow \R $, a root node $r\in V$, and \emph{penalties} $\X: V \rightarrow \R$. The goal is to find a tree that connects a subset of nodes to $r$, while trying to minimize the cost of the tree and the sum of the penalties of nodes $\I$ not connected to $r$.  Hence,
\[ f(\I,\X) = \sum_{i\in \I} X_i + \text{Min-Steiner-Tree}(V\setminus \I), 
\]
where $\text{Min-Steiner-Tree}(V\setminus \I)$ denotes the minimum cost tree connecting all nodes in $V\setminus \I$ to $r$.
\end{itemize}

We can now describe an abstract utility-maximization model that captures problems such as Pandora's box, max-weight matching, and max-spanning tree,  in a single unifying framework,

\noindent \textbf{Utility-Maximization}  ~  Suppose we are given a downward-closed (packing)\footnote{An independence family $\F \subseteq 2^V$ is called \emph{downward-closed} if $A\in \F$ implies $B\in \F$ for any $B\subseteq A$. A set-system is called \emph{upward-closed} if its complement is downward-closed.} constraint $\F \subseteq 2^V$ and a semiadditive function $\val$. Each element $i\in V$ takes a value $X_i$ independently from a known probability distribution. To find the  outcome $X_i$ we have to pay a known probing price $\price_i$. The goal is to adaptively probe a set of elements $\probed \subseteq V$ and select  $\I \subseteq \probed$ that is feasible (i.e., $\I \in \F$) to  maximize the   expected \emph{utility}, 
 \[ \E \left[ \val(\I,\X) - \sum_{i\in \probed} \price_i \right],
 \]
where the expectation is over random variables $\X$ and any internal randomness of the algorithm.

For example, in the max-weight matching problem $\val$ is an additive function and a subset of edges $\I$ is feasible if they form a matching. Similarly, when $\val$ is additive and $\F$ is a matroid, this framework captures max-weight matroid rank function, which contains Pandora's box and max-spanning tree as special cases.

The following is our main result for the utility-maximization problem: 

 \begin{theorem} \label{thm:utilityMax}
 For the utility-maximization problem for additive value functions and various  packing constraints $\F$,  we  obtain the following efficient algorithms.
\begin{itemize}
\item \emph{$k$-system}\footnote{An independence family $\F\subseteq 2^{V}$ is a \emph{$k$-system} if for any $Y\subseteq V$ we have $\frac{\max_{A\in \mathcal{B}(Y)}(|A|)}{\min_{A\in \mathcal{B}(Y)}(|A|)} \leq k$, where $\mathcal{B}(Y)$ denotes the set of maximal independent sets of $\F$ included in $Y$~\cite{CCPV-SICOMP11}.
%$A\in \F$ and any $e\in V$,  $\exists B\subseteq A$ of size at most $k$ s.t. $(A\setminus B) \cup e \in \F$. 
These are more general than intersection of $k$ matroids: e.g., a $2$-system captures matching in general graphs and a $k$-system captures matching in a hypergraph with edges of size at most $k$.}:
For stochastic element values, we get a $k$-approximation.
\item  \emph{Knapsack}: For stochastic item values and known item sizes, we get a  $2$-approximation.
\end{itemize}
\end{theorem}

Some important corollaries of Theorem~\ref{thm:utilityMax} are an optimal algorithm for the max-weight matroid rank  problem\footnote{For  weighted matroid rank functions, Kleinberg et al.~\cite{KWW-EC16}  independently obtain a similar result.} and a $2$-approximation algorithm for the max-weight matching problem. Theorem~\ref{thm:utilityMax}  is particularly interesting because it gives approximation results for mixed-sign objectives, which are usually difficult to handle. %~\cite{CDR-SODA06}.

We also show that if $\val$ is allowed to be any monotone submodular function then one \emph{cannot} obtain good approximation results: there is  an $\tilde{\Omega}(\sqrt{n})$  hardness even in a deterministic setting (see \ifFULL Section~\ref{sec:submodHardness} \else  full version\fi).

 Next, we  describe a  disutility-minimization model that captures problems like the \emph{min}-cost spanning tree.

\noindent \textbf{Disutility-Minimization }  ~
 Suppose we are given an upward-closed (covering) constraints $\F' \subseteq 2^V$ and a semiadditive function $\cost$. 
Each element $i\in V$ takes a value $X_i$ independently from a known probability distribution. To find the  outcome $X_i$ we have to pay a known probing price $\price_i$. The goal is to adaptively probe a set of elements $\probed \subseteq V$ and select  $\I \subseteq \probed$ that is feasible (i.e., $\I \in \F'$) to  minimize the  expected \emph{disutility}, 
 \[ \E \left[ \cost(\I,\X) + \sum_{i\in \probed} \price_i  \right],
 \]
where the expectation is over random variables $\X$ and any internal randomness of the algorithm.

For example, in the min-cost spanning tree problem, $\cost$ is an additive function and a subset of edges $\I$ are in $\F'$ if they contain a spanning tree. Similarly,  when $\val$ is the semiadditive facility location function as defined in Eq.~\eqref{eq:facilLocation} and  every non-empty subset of $V$ is feasible in $\F'$, this captures the min-cost facility location problem.

%\begin{rem}
\noindent  \textbf{\emph{Remark:}} The disutility-minimization problem can be also modeled as a utility-maximization problem by allowing item values to be negative and working with the infeasibility constraints (if $A\in \F'$ then $V\setminus A \in \F$), but such a transformation is not approximation factor preserving. 
%\end{rem}

We now mention our results in this model. (See \ifFULL Section~\ref{sec:applications} \else full version \fi for formal descriptions of these problems).

\begin{theorem} \label{thm:disutilityMinim} 
For the disutility-minimization problem for various  covering constraints $\F'$,  we  obtain the following efficient algorithms.
\begin{itemize}
\item \emph{Matroid Basis}: For stochastic element costs, we  get the optimal adaptive algorithm.
\item \emph{Set Cover}: For stochastic costs of the sets,  we  get a $\min\{ O(\log |V|), f\}$-approximation, where $V$ is the universe and $f$ is the maximum number of sets in which an element can occur.
\item \emph{Uncapacitated Facility Location}: For stochastic facility opening costs in a given metric, we  get a $1.861$-approximation.
\item \emph{Prize-Collecting Steiner Tree}: For stochastic penalties in a given graph with given edge costs, we get a $3$-approximation.
\item \emph{Feedback Vertex Set}: For stochastic vertex costs in a given graph we  get an $O(\log n)$-approximation.
\end{itemize}
\end{theorem}

%--------------------------------------------------
\subsection{Constrained Utility-Maximization}  \label{sec:introConstrainedUtilMax}
Our techniques can extend to settings where we impose  restrictions on the set of elements that we can probe. In particular, we are given a downward-closed set system $\J$  and the constraints allow us to only probe a subset of elements  $\probed \in \J$. This is different from the model discussed in Section~\ref{sec:modelResults} as earlier we could probe any set of elements but could get value  for only a subset elements that belong to $\F$. As an example, consider a generalization of the Pandora's box problem where besides paying probing prices, we can only probe at most $k$ elements. We now formally define our problem.

\noindent \textbf{Constrained Utility-Maximization }  ~
Suppose we are given  downward-closed probing constraints $\J \subseteq 2^V$ and   probability distributions of independent non-negative variables $X_i$ for $i \in V$. To find $X_i$ we have to pay a probing price $\price_i$. The goal is to  probe a set of elements $\probed \in \J$ to  maximize the   expected \emph{utility}, 
 \[ \E \left[ \max_{i\in \probed}\{X_i\} - \sum_{i\in \probed} \price_i \right].
 \]

\noindent \textbf{\emph{Remark:}}
%\begin{rem}
One can define an even more general version of this problem where we simultaneously have both downward-closed set systems $\F$ and $\J$, and the goal is to maximize a semiadditive function $\val$ corresponding to $\F$, while probing a set feasible in $\J$. For ease of exposition, we do not discuss it here and consider our value function  to be  the $\max$ function, as in the original Pandora's box  problem.
%\end{rem}

Depending on the  family of constraints $\J$,  we design efficient approximation algorithms for some settings of the above problem. The following is our main result for this problem (proof in Section~\ref{section:reducingConstUtilMaxToNonAdap}).

\begin{restatable}{theorem}{constrainedProbing} \label{thm:constrainedUM}
 If the constraints $\J$ form an $\ell$-system then the  constrained utility-maximization problem has  a $3(\ell+1)$-approximation algorithm. 
\end{restatable}
\noindent Since the cardinality (or any matroid)  constraint forms a $1$-system, an  application of Theorem~\ref{thm:constrainedUM} gives  a $6$-approximation algorithm for the Pandora's box problem under a cardinality probing constraint.

The above constrained utility-maximization problem is powerful and can  be used as a framework   to study variants of  Pandora's box. For example, consider the Pandora's box problem where we also allow to select an \emph{unprobed} box $i$ and get value $\E[X_i]$, without even paying its probing price $\price_i$. This  can be modeled using a partition matroid constraint where each box has two copies and the constraints  allow us to probe at most one of them. The first copy has a deterministic value $\E[X_i]$ with zero probing price and the second copy has a random value $X_i$ with price $\price_i$. Using Theorem~\ref{thm:constrainedUM}, we get a $6$-approximation for this variant.

%\noindent \textbf{Set-Probing }  ~
%\paragraph{Set-Probing} 
As a non-trivial  application of this constrained utility-maximization framework,  in Section~\ref{section:setProbing} we discuss a \emph{set-probing utility-maximization} problem where the costs are on subsets of random variables, instead of individual variables. Thus for a subset $S\subseteq V$, we pay price $\price_S$ to simultaneously probe all the random variables  $X_i$ for $i\in S$. This complicates the problem because  to find $X_i$, we can  probe a ``small" or a ``large" set  containing $i$, but at different prices. Formally, we define the problem as follows.

\noindent \textbf{Set-Probing Utility-Maximization}  ~
Given probability distributions of independent non-negative variables $X_i$ for $i \in V$ and 
given set family $\S=\{S_1, S_2, \ldots, S_m\}$, where $S_j \subseteq V$ for $j \in [m]$ has a probing price $\price_j \geq 0 $. The problem is to probe some of the sets in $\S$ with indices in $\probed \subseteq [m]$ to maximize 
\[ \E \left[ \max_{\exists j\in \probed ~s.t.~ S_j \ni i} \{ X_i\} - \sum_{j\in \probed} \price_j \right].
\]
Note that when we probe multiple sets containing an element $i$, we find the same value $X_i$ and not a fresh sample from the distribution.

\noindent \textbf{\emph{Remark:}}
%\begin{rem}
If the sets $S_j$ are pairwise disjoint then one can solve the above problem \emph{optimally}: replacing  each set $S_j$ with a new random variable  $X'_j = \max_{i\in S_j} \{X_i \}$ having probing price $\price_j$ reduces it to Pandora's box.
%\end{rem}

%(In the special case where all given sets are disjoint, this reduces to  Pandora's box problem.) 
We use the constrained utility-maximization problem framework to show the following result.
\begin{restatable}{theorem}{setProbing}\label{thm:setprobing}
The set-probing utility-maximization problem has a $3(\ell+1)$-approximation efficient algorithm, where $\ell$ is the size of the largest set in $\S$. Moreover, no efficient algorithm can be  $o(\ell/ \log \ell)$-approximation, unless $P=NP$.
\end{restatable}

%--------------------------------------------------
\subsection{Our Techniques}

How do we bound the utility/disutility of the optimal adaptive strategy?
The usual techniques  in approximation algorithms for stochastic problems (see related work in Section~\ref{sec:relatedWork})  either use a linear program (LP) to bound the optimal strategy, or directly argue about the adaptivity gap of the optimal  decision tree.
Neither of these techniques is helpful  because the natural LPs fail to capture a mixed-sign objective---they wildly overestimate the value of the optimal strategy. On the other hand, the adaptivity gap of our problems is large even for the special case of the Pandora's box problem---see an example in \ifFULL Section~\ref{sec:adapGapHardness}.\else  the full version.\fi

We need two crucial ideas for both our utility-maximization and  disutility-minimization  results. % The first idea is needed to bound the utility of the optimal adaptive  strategy. 
Our first idea is to show that for semiadditive  functions, one can bound the utility/disutility of the optimal strategy  in the {price-of-information} world (hereafter, the \emph{\POI world}) using a related instance in a world where there is no price to finding the parameters, i.e., $\price_i=0$ (hereafter, the \emph{\FreeInfo world}). This new instance still has independent random variables, however, the distributions are modified based on the original probing price $\price_i$ (see Defn~\ref{defn:Y}).  This proof crucially relies on the semiadditive nature of our value/cost function. %Next, we show how to design  strategies with expected utility/disutility close to this bound.

Our second idea is to show that any algorithm with ``nice" properties  in the \FreeInfo world can be used to get an algorithm with a similar expected utility/disutility in the \POI world. We call such a nice algorithm \frugal  (and define it formally in Section~\ref{section:frugalAlgoDefn}). For intuition, imagine a \frugal algorithm to be a greedy algorithm, or an algorithm that  is not ``wasteful"---it picks elements irrevocably. This also includes simple primal-dual algorithms that do not have the \emph{reverse-deletion} step.

\begin{restatable}{theorem}{frugalToAdaptive}\label{thm:frugalToAdaptive}
If there exists a \frugal  $\alpha$-approximation  Algorithm~\A to maximize (minimize) a semiadditive function over some packing constraints $\F$ (covering constraints $\F'$) in the \FreeInfo world then there exists an  $\alpha$-approximation algorithm for the corresponding utility-maximization (disutility-minimization) problem in the \POI world.
\end{restatable}

\noindent Finally, to prove our results from  Section~\ref{sec:modelResults}, in \ifFULL Section~\ref{sec:applications} \else the full version \fi we show why many classical algorithms, or their suitable modifications,  are \frugal. 

We remark that although  Theorem~\ref{thm:frugalToAdaptive} gives good guarantees for several combinatorial problems in the \POI world, there are some natural problems where there are no good \frugal algorithms. One such important problem is to find the shortest $s-t$ path in a given graph with stochastic edge lengths. It's an interesting open question to find some approximation guarantee for this problem. Another interesting open question is to show that finding the optimal policy for the max-weight matching in the \POI world is \emph{hard}.

Our techniques for the constrained utility-maximization problem in Section~\ref{sec:extensions} again use the idea of bounding  this problem in the \POI world with a similar problem in the \FreeInfo world. This latter problem turns out to be the same as the \emph{stochastic probing} problem studied  in~\cite{GN-IPCO13,ASW14,GNS-SODA16,GNS-SODA17}. By proving an extension of the adaptivity gap result of Gupta et al.~\cite{GNS-SODA17}, we show that  one can  further simplify these \FreeInfo problems to  non-adaptive utility-maximization problems (by losing a constant  factor). This we can now  (approximately) solve using our  techniques for the utility-maximization problem.

%Our techniques for the constrained utility-maximization problem in Section~\ref{sec:extensions} use an extension of the adaptivity gap result of Gupta et al.~\cite{GNS-SODA17}. We show that by losing a constant approximation factor we can reduce such problems to simpler non-adaptive utility-maximization problems, which we can now  approximately solve using the above discussed techniques. The application of this framework to the set-probing utility-maximization problem is based on a transformation that converts  this unconstrained problem into a constrained utility-maximization problem over some hypothetical $\ell$-system constraints.

%--------------------------------------------------

\subsection{Related Work}\label{sec:relatedWork}

An influential  work of Dean et al.~\cite{DGV-FOCS04}  considered the stochastic knapsack problem where we have stochastic knowledge about the sizes of the items. Chen et al.~\cite{CIKMR09} studied stochastic matchings where we find about an edge's existence only after probing, and Asadpour et al.~\cite{ANS-WINE08} studied stochastic submodular maximization where the items may or may not be present.  Several followup papers have appeared, e.g. for knapsack~\cite{BGK-SODA11,Ma-SODA14}, packing integer programs~\cite{DGV05,CIKMR09,BGLMNR-Algorithmica12}, budgeted multi-armed
bandits~\cite{GM-SODA07TALG12,GM-STOC07,GKMR-FOCS11,LiYuan-STOC13,Ma-SODA14}, orienteering~\cite{GuhaM09,GKNR-SODA12,BN-IPCO14}, matching~\cite{A11,BGLMNR-Algorithmica12,BCNSX15,AGM15}, and submodular objectives~\cite{GN-IPCO13,ASW14}. 
Most of these results proceed by showing that  the stochastic problem has a small adaptivity gap and  then focus on the non-adaptive  problem. In fact, Gupta et al.~\cite{GNS-SODA17,GNS-SODA16} show that  adaptivity gap for submodular functions over any  packing constraints is $O(1)$. 

Most of the above works do not capture mixed-sign objective of maximizing the value \emph{minus} the prices. Some of them instead model this as a knapsack constraint on the prices.  Moreover, most of them are for  maximization problems  as
for the minimization setting  even the non-adaptive problem of probing $k$ elements to minimize the expected minimum  value  has no polynomial approximation~\cite{GGM-TALG10}.  This is also the reason we do not consider constrained (covering) disutility-minimization in Section~\ref{sec:introConstrainedUtilMax}. There is also a large body of work in related models where information has a price. We refer the readers to the following papers and the references therein~\cite{GK-FOCS01,CFGKRS-Journal02,kannan2003selection,GMS-Transactions07,CJKBSK-AAAI15,AbbasH-Book15,CHHKK-COLT15}.

The Pandora's box solution can be written as a special case of the Gittins index theorem~\cite{GJ-Journal74}. Dumitriu et al.~\cite{DTW-SIDMA03} consider a  minimization variant of the Gittins index theorem when there is no discounting.
Another very relevant paper is that of Kleinberg et al.~\cite{KWW-EC16}, while their results are to design auctions. Their proof of the Pandora's box problem  inspired this work. %us to study designing  optimal/approximation algorithms for combinatorial  problems in the \POI world.

%\noindent \textbf{Organization}~
\paragraph{Organization} 
In Section~\ref{sec:boundingOptStrategy} we show how to bound the optimal strategy in the \POI world using a corresponding problem in the \FreeInfo world.  In Section~\ref{sec:frugalToAdaptive} we introduce the idea of using a \frugal algorithm to design a strategy with a good expected utility/disutility in the \POI world. In \ifFULL Section~\ref{sec:applications} \else the full version \fi we show why many classical algorithms, or their suitable modifications,  are \frugal.
Finally, in Section~\ref{sec:extensions} we discuss the settings where we  have  probing constraints,
%, its connections to the Stochastic Probing problem, 
and its application to the set-probing problem.

%%%%%%%%%%%%%%%%%%%%%%%%%%%%%%

\section{Bounding the Optimal Strategy for Utility/Disutility Optimization} \label{sec:boundingOptStrategy}
In this section we bound the expected utility/disutility of the optimal adaptive strategy for a combinatorial optimization in the \POI world   in terms of a surrogate problem in the  \FreeInfo world. 
We first define the \emph{grade} $\tau$ and \emph{surrogate} $Y$  of non-negative random variables.

\begin{defn} [\emph{Grade} $\tau$] \label{defn:tau}
For any non-negative random variable $X_i$, let $\taumax_i$ be the solution to equation $\E[(X_i- \taumax_i)^+] = \price_i$ and let $\taumin_i$ be the solution to equation $\E[(\taumin_i - X_i)^+] = \price_i$.
\end{defn}

\begin{defn} [\emph{Surrogate} $Y$] \label{defn:Y}
For any non-negative random variable $X_i$, let $\Ymax_i = \min\{ X_i, \taumax_i\}$ and let $\Ymin_i = \max\{ X_i,\taumin_i \}$.
\end{defn}

Note that $\taumax_i$ could  be negative in the above definition. The following lemmas bound  the optimal  strategy in the \POI world in terms of the optimal strategy of a surrogate problem in the \FreeInfo world. %terms of a problem with no mixed-sign objective.
\begin{lemma} \label{lem:boundAdapMax}
The expected utility of the optimal strategy to maximize a semiadditive  function $\val$ over packing constraints $\F$ in the \POI world is at most 
\[	\E_{\X} [\max_{\I \in \F} \{ \val(\I, {\bYmax)}\}].
\]
%where $Y_i = \min\{X_i, \tau_i \}$ and $\tau_i$ be the solution to equation $\E[(X_i- \tau_i)^+] = \price_i$.
\end{lemma}

\begin{lemma} \label{lem:boundAdapMin}
The expected disutility of the optimal strategy to minimize a semiadditive  function $\cost$ over covering constraints $\F'$ in the \POI world  is at least 
\[	\E_{\X} [\min_{\I \in \F'} \{ \cost(\I, \bYmin)\}].
\]
%where $Y_i = \min\{X_i, \tau_i \}$ and $\tau_i$ be the solution to equation $\E[(X_i- \tau_i)^+] = \price_i$.
\end{lemma}

\ifFULL We  only prove Lemma~\ref{lem:boundAdapMax} as the proof of Lemma~\ref{lem:boundAdapMin} is similar. The ideas in this proof are similar to that of Kleinberg et al.~\cite[Lemma~1]{KWW-EC16} to bound the optimal adaptive strategy for  Pandora's box.

\begin{proof}[Proof of Lemma~\ref{lem:boundAdapMax}] 
Consider a fixed optimal adaptive strategy.  Let $A_i$  denote the indicator variable that element $i$ is selected into $\I$ and  let $\one_i$  denote the indicator variable that element $i$ is probed by the optimal strategy. Note that these indicators are correlated and the set of elements with non-zero $A_i$ is feasible in $\F$. Now, the optimal strategy has expected utility
\begin{align*}
 &= \E \left[ \val(\I,\X) - \sum_{i\in \probed} \price_i \right] \\ %\qquad = \qquad  \sum_{i\in \I} X_i + \val_{\I}(\Z) \\
&= \E \left[\sum_i \left( A_i X_i - \one_i \price_i \right) \right] + \E[h(\I)] \\
& = \E\left[\sum_i \left( A_i X_i - \one_i \E_{X_i}[ (X_i - \taumax_i)^+] \right) \right] + \E[h(\I)],
\intertext{using the definition of $\price_i$. Since value of $X_i$ is independent of whether it's probed or not, we simplify to}
&= \E\left[\sum_i \left( A_i X_i - \one_i  (X_i - \taumax_i)^+ \right) \right]  + \E[h(\I)].
\intertext{Moreover, since we can select an element into $\I$ only after probing, we have $\one_i \geq A_i$.  This implies that the expected utility of the optimal strategy is}
&\leq \E\left[\sum_i \left( A_i X_i - A_i  (X_i - \taumax_i)^+ \right) \right] + \E[h(\I)]  \\
&  =  \E\left[ \sum_i A_i \Ymax_i \right]  + \E[h(\I)] \\
& = \E [ \val(\I,\bYmax) ].
\end{align*}
Finally, since elements in $\I$ form a feasible set, this is at most $\E [\max_{\I \in \F} \{ \val(\I, \bYmax)\}]$.
\end{proof}

\else
We discuss the proof of Lemma~\ref{lem:boundAdapMax} (proof of Lemma~\ref{lem:boundAdapMin} is similar) in the full version. The ideas in this proof are similar to that of Kleinberg et al.~\cite{KWW-EC16} to bound the optimal  strategy for  Pandora's box.
\fi

%%%%%%%%%%%%%%%%%%%%%%%%%%%%%%

\section{Designing an Adaptive Strategy for Utility/Disutility Optimization
} \label{sec:frugalToAdaptive}
In this section we introduce the notion of a \frugal algorithm and prove Theorem~\ref{thm:frugalToAdaptive}. We  need the following notation.
%and show how they can be used to design  algorithms for utility-maximization or disutility-minimization problems.  
\begin{defn}[ $\Ym$ ]
For any  vector $\Y$ with indices in $V$ and any $M\subseteq V$, let $\Ym$  denote  a vector of length $|V|$ with entries $Y_j$ for $j\in M$ and a symbol $*$, otherwise.
\end{defn}

%--------------------------------------------------------------------------------------------------------------------

\subsection{A \frugal Algorithm} \label{section:frugalAlgoDefn}
The notion of a \frugal algorithm is similar to that of a greedy algorithm, or any other algorithm that  is not ``wasteful"---it selects elements one-by-one and irrevocably.  Its definition captures ``non-greedy" algorithms such as the  primal-dual algorithm for set cover that does not have the \emph{reverse-deletion} step.

We  define a \frugal algorithm in the packing setting. 
Consider a packing problem in the \FreeInfo world (i.e.,  $\forall i, \price_i=0$ ) where we want to find a feasible set $\I \in \F$ and $\F \subseteq 2^V$ are some downward-closed  constraints, while trying to maximize a semiadditive function 
$\val(\I,\Y) = \sum_{i\in \I} Y_i + h(\I)$.

\begin{defn}[\frugal Packing  Algorithm]\label{defn:frugal}
For a packing problem with constraints $\F$ and value function $\val$, we say Algorithm \A is \frugal if  there exists a \emph{marginal-value}  function $g(\Y,i,y):\mathbb{R}^V \times V \times \R \rightarrow \R$ that is increasing in $y$, and for which the pseudocode is given by Algorithm~\ref{alg:frugalPacking}. We note that this algorithm always returns a feasible solution if we assume $\emptyset \in \F$.
\begin{algorithm} 
\caption{\frugal Packing Algorithm \A}
\label{alg:frugalPacking}
\begin{algorithmic}[1] 
\State Start with $M=\emptyset$ and  $v_i=0$ for each element $i \in V$.
\State For each element $i\not\in M$, compute  $v_i = g( \Ym, i,Y_i)$. Let $j = \argmax_{i\not \in M ~\&~ M\cup i \in \F} \{v_i\}$. \label{alg:FrugalcomputeVi}
\State If $v_j>0$ then add $j$ into $M$ and  go to Step~\ref{alg:FrugalcomputeVi}. Otherwise, return $M$.
\end{algorithmic}
\end{algorithm}
\end{defn}
\noindent A simple example of a  \frugal packing algorithm is the greedy algorithm to find the maximum weight spanning tree (or to  maximize any weighted matroid rank  function), where $ g( \Ym, i,Y_i) = Y_i$.

\ifFULL 
We similarly define a \frugal algorithm in the covering setting. Consider a covering problem in the \FreeInfo world where we want to find a feasible set $\I \in \F'$, where $\F' \subseteq 2^V$ is some upward-closed  constraint, while trying to minimize a semiadditive function 
$\cost(\I,\Y) = \sum_{i\in \I} Y_i + h(\I)$.

\begin{defn}[\frugal Covering  Algorithm]\label{defn:frugalCovering}
For a covering problem with constraints $\F'$ and cost function $\cost$, we say Algorithm \A is \frugal if  there exists a \emph{marginal-value}  function $g(\Y,i,y):\R^V \times V \times \R \rightarrow \R$ that is increasing in $y$, and  for which  the pseudocode is given by Algorithm~\ref{alg:frugalCoverage}.
\begin{algorithm} 
\caption{\frugal Coverage Algorithm \A}
\label{alg:frugalCoverage}
\begin{algorithmic}[1] 
\State Start with $M=\emptyset$ and  $v_i=0$ for each element $i \in V$.
\State For each element $i\not\in M$, compute  $v_i = g( \Ym, i,Y_i)$. Let $j = \argmax_{i\not \in M} \{v_i\}$. \label{alg:FrugalcomputeViCoverage}
\State If $v_j>0$ then add $j$ into $M$ and  go to Step~\ref{alg:FrugalcomputeViCoverage}. Otherwise, return $M$.
\end{algorithmic}
\end{algorithm}

\noindent We note that for a covering problem it is unclear whether Algorithm~\ref{alg:frugalCoverage} returns a feasible solution as we do not appear to be looking at our covering constraints $\F'$. To overcome this, we say the marginal-value function $g$ \emph{encodes} $\F'$ if whenever $M$ is infeasible then there exists an element $i\not\in M$ with $v_i>0$. This means that the algorithm will return a feasible solution as long as $V \in \F'$.
\end{defn}

\noindent A simple example of a  \frugal covering algorithm is the greedy min-cost set cover algorithm, where $ g( \Ym, i,Y_i) = \left( {|\bigcup_{j\in M \cup i} S_j| - | \bigcup_{j\in M} S_j|} \right)/{Y_i}$. Note that here $g$ \emph{encodes} our coverage constraints.

\noindent \textbf{\emph{Remark:}} 
%\begin{rem} 
Observe that a  crucial difference between \frugal packing and covering  algorithms is that a \frugal packing  algorithm has to  handle   $\Y \in \mathbb{R}^V$ (i.e. some entries in $\Y$ could be negative) but a \frugal covering  algorithm has to only handle $\Y \in \R^V$. The intuition behind this difference is that unlike the disutility minimization problem, the utility maximization problem has a mixed-sign objective.
%\end{rem}

\else
In the full version we similarly define a \frugal algorithm in the covering setting. 
\fi

%------------------------------------------------------------------

\subsection{Using a {\frugal} Algorithm to Design an Adaptive Strategy}
After defining the notion of a \frugal algorithm, we can now prove Theorem~\ref{thm:frugalToAdaptive} (restated below). 
\frugalToAdaptive*
%\begin{theorem}\label{thm:frugalToAdaptive} If there exists a \frugal  $\alpha$-approximation  Algorithm~\A to maximize (minimize) a semiadditive function over some packing constraints $\F$ (covering constraints $\F'$) in the \FreeInfo world then there exists an  $\alpha$-approximation algorithm for the corresponding utility-maximization (disutility-minimization) problem in the \POI world. \end{theorem}

We prove  Theorem~\ref{thm:frugalToAdaptive} only for the utility-maximization setting as the other proof  is similar. Lemma~\ref{lem:boundAdapMax} already gives us an upper bound on the expected utility of the optimal strategy for the utility-maximization problem in terms of the expected value of a problem in the \FreeInfo world. This \FreeInfo problem can be solved using Algorithm~\A. The main idea in the proof of this theorem is to show that if Algorithm~\A is \frugal then we can also run a modified version of \A in the \POI world and get the same expected utility.

\begin{proof}[Proof of Theorem~\ref{thm:frugalToAdaptive}]
Let $Alg(\bYmax,\A)$ denote the set $\I \in \F$ returned by Algorithm \A when it runs with element weights $\bYmax$. Since \A is an $\alpha$-approximation algorithm (where $\alpha \geq 1$), we know
\begin{align} \label{eq:AlgAGuarantee}
 \val(Alg(\bYmax,\A), \bYmax)  \geq \frac{1}{\alpha} \cdot \max_{\I \in \F} \{ \val(\I, \bYmax) \}.
\end{align}

The following crucial lemma shows that one can design an adaptive strategy in the \POI world with the same expected utility.%, $ \val(Alg(\bYmax,\A), \bYmax)  $.

\begin{lemma} \label{lem:covertFrugalAlg}
If Algorithm~\A is \frugal then there exists an algorithm in the \POI world with expected utility
\[ \E_{\X} [\val(Alg(\bYmax,\A), \bYmax)].
\]
\end{lemma}

Before proving  Lemma~\ref{lem:covertFrugalAlg}, we finish the proof of Theorem~\ref{thm:frugalToAdaptive}. Recollect that Lemma~\ref{lem:boundAdapMax} shows that $\E[\max_{\I \in \F} \{ \val(\I, \bYmax) ]$ is an upper bound on the expected optimal utility in the \POI world. Combining  this with Lemma~\ref{lem:covertFrugalAlg} and Eq.~\eqref{eq:AlgAGuarantee} gives an $\alpha$-approximation algorithm in the \POI world.
\end{proof}

We first give some intuition for the missing Lemma~\ref{lem:covertFrugalAlg}.
The lemma is surprising because it says that  there exists an algorithm in the \POI world that has the same expected utility as Algorithm~\A in the \FreeInfo world, where there are no prices. The fact that in the \FreeInfo world Algorithm \A can only get the smaller surrogate values $\Ymax_i = \min \{X_i , \taumax_i\}$, instead of the actual value $X_i$, comes to our rescue. We show that  $\Ymax$  is defined in a manner to balance this difference in the values with the probing prices.

\begin{proof}[Proof of Lemma~\ref{lem:covertFrugalAlg}]
Since \A is \frugal, we would like to run Algorithm~\ref{alg:frugalPacking} in the \POI world. The difficulty is that we do not know $\bYmax$ values of the unprobed elements. To overcome this hurdle, consider Algorithm~\ref{alg:frugalToAdaptive} that uses the grade $\taumax$  as a proxy for $\bYmax$ values of the unprobed elements. 
%We make the following observation.

\begin{claim} \label{claim:sameSolution}The set of elements returned by Algorithm~\ref{alg:frugalToAdaptive} is the same as that by Algorithm~\ref{alg:frugalPacking} running with $\Y=\bYmax$. 
\end{claim}

%\begin{claim} \label{claim:sameSolution} When executed with $\Y=\bYmax$,  Algorithm~\ref{alg:frugalToAdaptive} and Algorithm~\ref{alg:frugalPacking} return the same set of elements.
%\end{claim}

\ifFULL
\begin{proof}[Proof of Claim~\ref{claim:sameSolution}]
\noindent We prove the claim by induction on the number of elements selected by Algorithm~\ref{alg:frugalToAdaptive}. Suppose the set of elements  selected by both the algorithms into  $M$  are the same till now and  Algorithm~\ref{alg:frugalToAdaptive} decides to select  element $j$ in Step~\ref{algStep:utilMaxPickup}(a).
 This means that $j$ is already probed before this step. The only concern is that Algorithm~\ref{alg:frugalToAdaptive} selects $j$ without probing some other element $i$ based on  its  grade  $\taumax_i$. We observe that this step is consistent with Algorithm~\ref{alg:frugalPacking} because  $\Ymax_i \leq \taumax_i$ and $g( \Ymmax, i,\Ymax_i)$ is an increasing function in $\Ymax_i$, which implies  
\[ g( \Ymmax, i,\Ymax_i) \quad \leq \quad g(\Ymmax,i, \taumax_i) \quad  \leq \quad g( \Ymmax, i,\Ymax_j). \qedhere \] 
\end{proof}
\else
\noindent See proof in the full version.
\fi
\noindent An immediate corollary is that   value of Algorithm~\ref{alg:frugalToAdaptive} in the \FreeInfo world is  
\begin{align}
\E_{\X} [\val(Alg(\bYmax,\A), \bYmax)] .
\end{align}
In  Claim~\ref{claim:POIutility} we argue that this expression also gives  expected utility of Algorithm~\ref{alg:frugalToAdaptive} in the \POI world, which completes the proof of Lemma~\ref{lem:covertFrugalAlg}.
\end{proof}

\begin{claim} \label{claim:POIutility} The expected utility of Algorithm~\ref{alg:frugalToAdaptive} in \POI world is
\[ \E_{\X} [\val(Alg(\bYmax,\A), \bYmax)].
\]
\end{claim}
%\noindent Finally, Claim~\ref{claim:POIutility} proves Lemma~\ref{lem:covertFrugalAlg} as corollary.

\begin{proof}[Proof of Claim~\ref{claim:POIutility}] 
We first expand the claimed expression,
\begin{align} \label{eq:valueOfAlg}
\E_{\X} [\val(Alg(\bYmax,\A), \bYmax)] = \E\left[ \sum_{i\in Alg(\bYmax,\A)} \Ymax_i \right] + \E[h({Alg(\bYmax,\A)}) ].
\end{align}
Observe that to prove the claim we can ignore the  second term, $\E[h({Alg(\bYmax,\A)}) ]$, because it contributes the same  in both the worlds (it is only a function of the returned feasible set). We now argue that in every step of Algorithm~\ref{alg:frugalToAdaptive}  the expected change in $\sum_{i\in M} \Ymax_i$ in  the \FreeInfo world is the same as the expected increase in  $\sum_{i\in M} X_i$ minus the probing prices
 in the \POI world.  %\snote{Clarify whether you are coupling or what!}

\begin{algorithm} 
\caption{Utility-Maximization}
\label{alg:frugalToAdaptive}
\begin{algorithmic}[1] 
\State Start with $M=\emptyset$ and  $v_i=0$ for all elements $i$.
\State For each element $i\not\in M$:  \label{alg:computeVi}
\Statex (a) if $i$ is probed let $v_i = g( \Ymmax , i, \Ymax_i)$.
\Statex (b) if $i$ is unprobed let $v_i = g(\Ymmax , i, \taumax_i)$.
\State Consider the element $j = \argmax_{i\not \in M~\&~ M\cup i \in \F} \{v_i\}$ and $v_j>0$. \label{algStep:frugalToAdaptiveSelection}
\Statex (a) If $j$ is already probed then select it into $M$ and set $v_j=0$. \label{algStep:utilMaxPickup}
\Statex (b) If $j$ is not probed then probe it. If $X_j\geq \taumax_j$ then select $j$ into $M$ and set $v_j=0$.
\State If every element $i\not\in M$ has $v_i=0$ then return set $M$. Else, go to Step~\ref{alg:computeVi}.
\end{algorithmic}
\end{algorithm}

We first consider the case when the next highest element $j$  in Step~\ref{algStep:frugalToAdaptiveSelection} of Algorithm~\ref{alg:frugalToAdaptive} is already probed and has $v_j >0$. 
In this case, the algorithm selects element $j$. Since this element has been already probed before (but not selected then), it means $X_{j}< \taumax_{j}$ and $X_{j}= \Ymax_{j}$. Hence the increase in the value of the algorithm in both the worlds is $X_j$.

Next, consider the case that the next highest element $j$ in Step~\ref{algStep:frugalToAdaptiveSelection} has not been probed before. Let $\mu_{j}$ denote the probability density function of random variable $X_{j}$.
Now the expected increase in the value  in the \FreeInfo world is
\[ \taumax_{j} \cdot \Pr[X_{j} \geq \taumax_{j}] \quad  = \quad  \taumax_{j} \cdot \int_{t= \taumax_{j}}^{\infty} \mu_{j}(t) dt  .
\]
This is because the algorithm selects this element in this step only if its value is at least $\taumax_{j}$, in which case $\Ymax_{j}= \taumax_{j}$. On the other hand, the expected increase in the value in the \POI world is given by
\[ -\price_{j} + \int_{t=\taumax_{j}}^{\infty} t\cdot \mu_{j}(t) dt.
\]
This is because we pay the probing cost $\price_{j}$ and get a positive value $X_{j}$ only when $X_{j} \geq \taumax_{j}$. Now using the definition of $\taumax_{j}$, we can simplify the above equation to
\[ - \int_{t=\taumax_{j}}^{\infty} (t-\taumax_{j}) \mu_{j}(t) dt + \int_{t=\taumax_{j}}^{\infty} t\cdot \mu_{j}(t) dt \quad = \quad  \taumax_{j} \cdot \int_{t= \taumax_{j}}^{\infty} \mu_{j}(t) dt .
\]
This shows that in every step of the algorithm the expected increase in the value in both the worlds is the same, thereby proving Claim~\ref{claim:POIutility}. 
\end{proof}

%%%%%%%%%%%%%%%%%%%%%%%%%%%%%%

\ifFULL
%%%%%%%%%%%%%%%%%%%%%%%%%%%%%%

\section{Applications to Utility/Disutility Optimization}\label{sec:applications}
In this section we show that for several combinatorial  problems there exist \frugal algorithms. Hence we can use Theorem~\ref{thm:frugalToAdaptive}  to obtain optimal/approximation algorithms for the corresponding utility-maximization or disutility-minimization problem in \POI world. 

%------------------------------------------------------------
\subsection{Utility-Maximization}
To recollect, in the utility-maximization setting we are given a semiadditive value function $\val \geq 0$ and a packing constraint $\F$. Our goal is to probe a set of elements $\probed$ and select a feasible set $\I \subseteq \probed$ in $\F$ to maximize expected utility, 
 \[ \E\left[ \val(\I,\X) - \sum_{i\in \probed} \price_i \right].
 \]
 We   assume that $\emptyset \in \F$ and hence there always exist a solution of utility zero.

%------------------------------------------------------------
\subsubsection{$k$-System}
Let $\val(\I,\X) = \sum_{i\in \I} X_i$ be an additive function and let $\F$ denote a $k$-system constraint in this setting. To prove Theorem~\ref{thm:utilityMax}, we observe that the greedy algorithm that starts with an empty set and at every step selects the next feasible element maximum marginal-value is an $\alpha$-approximation algorithm is a \frugal algorithm as defined in Defn~\ref{defn:frugal}. We know that this greedy algorithm is a $k$-approximation for additive functions over a $k$-system in  \FreeInfo world~\cite{Jenkyns-76,KH-DM78}. Hence, Theorem~\ref{thm:frugalToAdaptive} combined with the greedy algorithm gives the $k$-system part of Theorem~\ref{thm:utilityMax} as a corollary.

%------------------------------------------------------------
%\IGNORE{
\subsubsection{Knapsack}

Given a knapsack of size $B$, suppose each item $i$ has a known size $s_i$ ($\leq B$) but a stochastic value $X_i$. To find $X_i$, we have to  pay probing price $\price_i$. The goal is to probe a subset of items $\probed$ and select a subset $\I \subseteq \probed$, where $\sum_{i\in \I} s_i \leq B$, to  maximize the expected utility
\[	\E\left[ \sum_{i\in \I} X_i - \sum_{i\in \probed} \price_i \right].
\]

\noindent We can model this problem in our utility-maximization framework by taking $\val(\I,\X) = \sum_{i\in \I} X_i$ and $\F$ to contain every subset $S$ of items that fit into the knapsack. 

In the \FreeInfo world, consider a greedy algorithm that sorts items in decreasing order based on the ratio of their value and size, and then selects items greedily in this order until the knapsack is full. This greedy algorithm does not always give a constant approximation to the knapsack problem. Similarly, an algorithm that  selects only the most valuable item is not always a constant approximation algorithm (recollect that we can pick every item $i$ because  $s_i \leq B$). However, it's known that for any knapsack  instance if we randomly run one of the previous two algorithms, each w.p. half, then this is  a $2$-approximation algorithm.
 
From Theorem~\ref{thm:frugalToAdaptive}, we can simulate the greedy algorithm in the \POI world. Also, using the solution to the Pandora's box problem, we can simulate selecting the most valuable item in the \POI world. Hence, consider an algorithm that for any given knapsack problem in the \POI world, runs  either the simulated greedy algorithm or the Pandora's box solution, each with probability half. Such an algorithm is a $2$-approximation to the knapsack problem in the \POI world. 

%------------------------------------------------------------
\subsection{Disutility-Minimization}
To recollect, in the disutility-minimization setting we are given a semiadditive cost function $\cost \geq 0$ and a covering constraint $\F'$. Our goal is to probe a set of elements $\probed$ and select a feasible set $\I \subseteq \probed$ in $\F'$ to minimize expected disutility, 
 \[ \E\left[ \cost(\I,\X) + \sum_{i\in \probed} \price_i \right].
 \]

\noindent We will assume that $V \in \F$, and hence there always exists a feasible solution.

\subsubsection{Matroid Basis}

Given a matroid $\M$ of rank $r$ on $n$ elements, we consider the additive function $\cost(\I,\X) = \sum_{i\in \I} X_i$ and let $\F'$ be subsets of elements that contain a basis of $\M$.
To ensure that a feasible set of finite value exists, we make the following assumption. 
\begin{assumption}\label{assump:matroidBasisExists}
We can always extend a set $\I \in \M$ to a basis by probing and selecting items with zero penalty but with large probing cost $\price_0$. Thus it incurs an additional penalty of $(r -|rank(\I)|)\cdot \price_0$.
\end{assumption}

To use Theorem~\ref{thm:frugalToAdaptive}, we notice that the greedy algorithm that always selects the minimum cost independent element is \frugal. This is true because we choose marginal-value function $g$ to be the reciprocal of the weight of an element in Defn~\ref{defn:frugalCovering}. Now since the greedy algorithm is optimal for min-cost matroid basis, this proves the first part of Theorem~\ref{thm:disutilityMinim}.

%--------------------------------------------------------------------------------------------

\subsubsection{Set Cover}
Consider a problem where we are given sets $S_1, \ldots, S_m \subseteq V$ that have some unknown stochastic costs $X_i$. The goal is to select a set cover with minimum disutility, which is the sum of the set cover solution costs and the probing prices. We can model this problem in our framework by considering a the additive function $\cost(\I,\X) = \sum_{i\in \I} X_i$ and $\F'$ be set covers of $V$.

 To ensure that the solution is always bounded, we make the following assumption. 
\begin{assumption}\label{assump:setCoverExists}
There exists $ S_0 = [n]$ that covers all  elements and has  $X_0=0$, but a finite large  $\price_0$.
\end{assumption}

To prove the set cover part of Theorem~\ref{thm:disutilityMinim}, we first notice that  the classical $O(\log |V|)$ greedy algorithm for the min-cost set cover problem is \frugal. This is because the marginal-value function $g( \Ym, i, Y_i)$ in Defn~\ref{defn:frugalCovering} is equal to ${\left({|\bigcup_{j\in M \cup i} S_j|  - |\bigcup_{j\in M} S_j|}\right)}/{Y_i}$.

Next we give an $f$-approximation algorithm, where $f$ is the maximum number of sets in which an element can appear. We observe that the primal-dual $f$-approximation algorithm (see pseudocode in Algorithm~\ref{alg:setcoverPrimalDual}) for the min-cost set cover~\cite{BE-JAL81,WS-Book11} is also \frugal. This is because we can encode the information about the order $\sigma$ and the dual variables $y_j$ for $j\in M$ in  the  marginal-value function $g( \Ym, i,Y_i)$ in Defn~\ref{defn:frugalCovering}.

\begin{algorithm} 
\caption{Primal-dual algorithm for min-cost set cover}
\label{alg:setcoverPrimalDual}
\begin{algorithmic}[1] 
\State Fix an order $\sigma$ on the ground elements. Start with $M=\emptyset$ and  $y_j=0$ for every ground element $j$.
\State Select the next element $j \not \in \bigcup_{i\in M} S_i$ according to $\sigma$ and raise its dual variable $y_j$ until  some set $i$ becomes tight, i.e., $\sum_{j\in S_i} y_j = Y_i$.
\label{algStep:primalDual}
\State Select every tight set into $M$.
\State If every ground element is covered in $\bigcup_{i\in M} S_i$ then return $M$, else go to Step~\ref{algStep:primalDual}.
\end{algorithmic}
\end{algorithm}

%------------------------------------------------------------

\subsubsection{Uncapacitated Facility Location}
Consider an uncapacitated facility location problem where we are given a graph $G=(V,E)$ with metric $(V,d)$ and  $\Clients\subseteq V$, however, facility opening costs $X_i$ for $i\in V$ are stochastic and can be found by paying a probing price $\price_i$. The goal is to probe a set of facility locations $\probed \subseteq V$ and open a non-empty subset $ I \subseteq \probed $  to minimize expected  disutility 
\[ 	\E\left[ \sum_{u\in \Clients} d(u,\I) + \sum_{i\in \I } X_i + \sum_{i\in \probed} \price_i \right],
\]
where $d(u,\I) = \min_{i\in \I} d(u,i)$.

We model the above problem in our framework by defining exponential number of elements that are indexed by $(i,S)$, for $i \in V$ and $S\subseteq \Clients$, which denotes that facility $i$ will serve clients $S$. Any subset of elements, say $\I = \{(i_1,S_1),(i_2,S_2), \ldots \}$, is feasible if the union of their clients covers $\Clients$. The semiadditive $\cost(\I,\X)$ is given by $\sum_{(i,S)\in \I} \left( X_i + \sum_{j\in S} d(i,j) \right)$.

We notice that the $1.861$-approximation greedy algorithm of  Jain et al.~\cite{JMMSV-JACM03}  for the  uncapacitated facility location problem is \frugal. In each step, their algorithm selects the next best element with minimum cost per client, where already opened facilities now have zero opening costs. The reciprocal of this value gives the marginal-value function $g$.
Hence, we can use Theorem~\ref{thm:frugalToAdaptive} to obtain a  $1.861$-approximation strategy. 

%------------------------------------------------------------
\subsubsection{Prize-Collecting Steiner Tree}% or Traveling Salesman}

Consider a Prize-Collecting Steiner tree problem (PCST)  where we are given a graph $G=(V,E)$ with some edge costs $\cc :E \rightarrow \R $,  a root node $r\in V$, and probability distributions on the independent penalties $X_i$ for $i\in V$. The stochastic penalties $X_i$ can be found by paying a probing price $\price_i$. 
The goal is to probe a set of nodes $\probed \subseteq V\setminus \{r\}$ and select a subset $ I \subseteq \probed $  to minimize expected  disutility, 
\[ \E \left[ \sum_{i\in \I} X_i + \text{Min-Steiner-Tree}(V\setminus \I) +\sum_{i\in \probed} \price_i \right], 
\]
where $\text{Min-Steiner-Tree}(V\setminus \I)$ denotes the minimum cost Steiner tree connecting all nodes in $V\setminus \I$ to $r$.

As discussed in Section~\ref{sec:modelResults}, we can model the PCST  in our disutility-minimization framework by noticing that the function $\cost(\X,\I) = \sum_{i\in \I} X_i + \text{Min-Steiner-Tree}(V\setminus \I)$ is semiadditive. We show that although the $2$-approximation Goemans-Williamson~\cite{GW-SICOMP95}   algorithm (hereafter, GW-algorithm) for PCST  in the \FreeInfo world is not \frugal, it can be modified to obtain a $3$-approximation \frugal algorithm for PCST. Combining this with Theorem~\ref{thm:frugalToAdaptive} gives a $3$-approximation algorithm for PCST in the \POI world.

We quickly recollect the $2$-approximation primal-dual GW-algorithm. (We do not repeat their  proof and refer  to~\cite[Chapter 14]{WS-Book11} for  details.) Their algorithm starts by making each node $i \in V\setminus \{r\}$ active with initial charge $p({\{i\}})  = X_i$. At any time, the algorithm grows a \emph{moat} around each active component $C$ and discharges $C$ at the same rate.  If a component $C$ runs out of charge, we make it inactive and mark every unlabeled node in the component with label $C$. If an edge $e$ becomes \emph{tight}, we pick $e$, merge the two components $C,C'$ connected by $e$, make both $C,C'$ inactive, and make $C\cup C'$ active with an initial charge of $p(C) + p(C')$. Any component that hits the component containing $r$ is made inactive. In the \emph{cleanup phase} we remove all edges that do not disconnect an unmarked node from $r$, while ensuring that if a component with label $C$ is connected to $r$ then every node with label $C' \supseteq C$ is also connected to $r$.

We first observe that the  GW-algorithm is not  \frugal. This is because whenever a node $i$ is labeled with a component $C \ni i$, the algorithm looks at the penalty $X_i$; however, the decision of whether to select $i$ into $\I$ (i.e., not connecting $i$ to $r$)  is not made until the cleanup phase. The reason is  that  some other active component $C'$ might later come and merge with $C$, and eventually connect $i$ to $r$. To fix this, we modify this  algorithm to make it \frugal. The idea is to immediately include the labeled vertices into $\I$.

Consider an algorithm that creates the same tree as the GW-algorithm; however, any node that ever gets labeled during the run of the algorithm is imagined to be included  into $\I$. This means that although our final tree might connect a labeled node $i$ to $r$, our algorithm still pays its penalty $X_i$. We argue that these additional penalties are at most the optimal PCST solution in the \FreeInfo world, which gives us a $3$-approximation \frugal algorithm. 

Finally, to argue that the additional penalties are not large, consider the state of the GW-algorithm before the cleanup phase. Let $\C$ denote the set of maximal inactive components. Clearly, each node $i$ that was every labeled belongs to some maximal tight component $C \in \C$. Hence, the sum of the additional penalties is upper bounded by $\sum_{ C \in \C} \sum_{i\in C} X_i$. 
Since each   component $C \in \C$ is tight, we know $\sum_{ C \in \C} \sum_{i\in C} X_i = \sum_{C \in \C} \sum_{S\subseteq C} y_S$, where $y_S$ are the dual variables corresponding to the moats. Since the dual solutions form a feasible dual-solution, they are a lower bound on the optimal solution for the problem. This proves that the additional penalty paid by our \frugal algorithm  in comparison to GW-algorithm is at most the optimal solution.

%------------------------------------------------------------
\subsubsection{Feedback Vertex Set}
Given an undirected graph $G=(V,E)$, suppose each node $i\in V$ has a stochastic weight $X_i$, which we can find by probing and paying price $\price_i$. The problem is to probe a set $\probed \subseteq V$ and select a subset $\I \subseteq \probed$ s.t. the induced graph $G[V\setminus \I]$ contains no cycle, while minimizing the expected disutility
\[ \E\left[ \sum_{i\in \I} X_i + \sum_{i\in \probed} \price_i \right].
\]

The above problem can be modeled in our framework by considering the additive function $\cost(\I,\X) = \sum_{i\in \I} X_i$ and $\F'$ contains a set of nodes $S$ if $G[V\setminus S]$ has no cycle. Becker and Geiger~\cite{BG-AI96} showed that the greedy Algorithm~\ref{alg:feedVertexSet} is an 
$O(\log n)$-approximation algorithm for the feedback vertex set problem in the \FreeInfo world. Since this algorithm is  \frugal, using Theorem~\ref{thm:frugalToAdaptive} we  get  an $O(\log n)$-approximation algorithm for minimizing disutility for the feedback vertex set problem in the \POI world.

\begin{algorithm} 
\caption{Greedy Algorithm for Feedback Vertex Set}
\label{alg:feedVertexSet}

\IGNORE{\begin{algorithmic}[1] 
\State While $G$ is not empty
\State Repeatedly remove degree $0$ or degree $1$ vertices from $G$, and let $G$ be the induced subgraph on the remaining vertices.
\State 
\end{algorithmic}
\end{algorithm}}

\begin{algorithmic}[1] 
\State Start with $R=M=\emptyset$ and  $v_i=0$ for each element $i \in V$.
\State  While $\exists i \in V\setminus (R\cup M)$ s.t. degree of $i$ in $G[V\setminus (R\cup M)]$ is $0$ or $1$, add $i$ to $R$. \label{algstep:feedVertSet}
\State For each element $i\not\in R\cup M$, compute  $v_i = degree(i,G[V\setminus (R\cup M)])/w(i)$, where $degree(i,G)$ is the degree of vertex $i$ in $G$ and $w(i)$ is the weight of vertex $i$.
\State  Let $j = \argmax_{i\not \in R\cup M} \{v_i\}$. Add $j$ to $M$. 
\State If $R \cup M \neq V$, go to Step~\ref{algstep:feedVertSet}. Otherwise, return $M$.
% If $v_j>0$ then add $j$ into $M$ and  go to Step~\ref{alg:FrugalcomputeViCoverage}. Otherwise, return $M$.
\end{algorithmic}
\end{algorithm}

%\begin{rem}
\noindent  \textbf{\emph{Remark:}}
The  $O(\log n)$-approximation   primal-dual algorithm in~\cite{BGNR-SICOMP98} (or in Chapter 7.2 of~\cite{WS-Book11}) can be also shown to be \frugal. This gives  another $O(\log n)$-approximation algorithm for minimizing disutility for feedback vertex set problem.
%\end{rem}

%------------------------------------------------------------
\IGNORE{
\subsubsection{Min-Cut}

Given an undirected graph $G=(V,E)$, a source node $s\in V$, and a sink node $t \in V$. Suppose each edge $i\in E$  has a stochastic cost $X_i$,  which we can find by probing and paying price $\price_i$. The problem is to probe a set $\probed \subseteq E$ and select a subset $\I \subseteq \probed$ s.t. the  graph $G(V,E\setminus \I)$ contains no path from $s$ to $t$, while minimizing the expected disutility
\[ \E\left[ \sum_{i\in \I} X_i + \sum_{i\in \probed} \price_i \right].
\]

The above problem can be modeled in our framework by considering the additive function $\cost(\I,\X) = \sum_{i\in \I} X_i$ and $\F'$ contains a set of edges $S$ if $G(V,E\setminus S)$ contains no path from $s$ to $t$.

}

%------------------------------------------------------------
\IGNORE{
\subsection{Shortest Path Tree}
Given a graph $G=(V,R)$, a root $r\in V$, and stochastic edge weights $X_i$ that we can learn by probing and paying cost $c_i$. The goal is to probe a set $S\subseteq E$ of edges and  find a tree $T \subseteq S$ rooted at $r$ s.t. for every vertex $u\in V$, length of path $r-u$ in $T$ gives the shortest path in $G$ with edge weights $X_i$, while minimizing $\sum_{i\in \probed} c_i$.

\begin{theorem} The greedy algorithm  can be modified to give an optimal algorithm for minimizing disutility for finding shortest path tree.
\end{theorem}
}
\else
\fi

%%%%%%%%%%%%%%%%%%%%%%%%%%%%%%%%

\section{Constrained Utility-Maximization}\label{sec:extensions}

In this section we consider a generalization of the Pandora's box  problem where we have an additional constraint that allows us to only probe a subset of elements  $\probed \subseteq V$ that belongs to a downward-closed constraint $\J$ (e.g., a cardinality constraint allowing us to probe at most $k$ elements). We restate our main result for the constrained  utility-maximization problem (see   problem definition in Section~\ref{sec:introConstrainedUtilMax}).% is Theorem~\ref{thm:constrainedUM}. 
\constrainedProbing*

Similar to Section~\ref{sec:boundingOptStrategy}, our strategy to prove Theorem~\ref{thm:constrainedUM} is to bound the constrained utility-maximization problem in the \POI world (a mixed-sign objective function) with a surrogate constrained utility-maximization problem in the \FreeInfo world (i.e., where $\price_i=0$ for all $i\in V$). This latter problem turns out 
to be the same as the \emph{stochastic probing} problem, which we define below in the form that is relevant to this paper.

\noindent \textbf{Stochastic Probing}  ~ 
Given downward-closed probing constraints  $\J \subseteq 2^V$ and probability distributions of independent non-negative variables $Y_i$ for $i \in V$, the stochastic probing problem is to \emph{adaptively} probe a subset $\probed \in \J$ to maximize the expected value $\E[\max_{i\in \probed} \{Y_i\}]$. Here, \emph{adaptively} means that the decision to probe which  element next can depend on the outcomes of the already probed elements.

%The above Stochastic Probing problem can be generalized by considering arbitrary monotone submodular functions 

%--------------------------------------------------------
\subsection{Reducing Constrained Utility-Maximization to Non-adaptive Stochastic Probing}\label{section:reducingConstUtilMaxToNonAdap}

The following lemma  bounds the expected utility of the constrained utility-maximization problem in the \POI world by the expected value of a stochastic probing problem in the \FreeInfo world.

\begin{restatable}{lemma}{boundOptConstrained} \label{lem:boundOptConstrained}
The expected utility of the optimal strategy for the constrained utility-maximization problem is  at  most the expected value of the optimal adaptive strategy for a  stochastic probing problem with the same constraints $\J$ and where the  random variables $Y_i$ for  $i\in V$ have probability distributions  $\Ymax_i$ (recollect, Defn~\ref{defn:Y}).
\end{restatable}
\ifFULL
\begin{proof}[Proof of Lemma~\ref{lem:boundOptConstrained}]
We start by noticing that the optimal strategy for our problem is given by a decision tree $T$ with leaves $\leaf$. For any root leaf path $P_{\leaf}$, the value of the optimal strategy is $\max_{i \in P_{\leaf}} \{X_i\} - \price(P_{\leaf})$. Thus the expected value of the optimal strategy is
\begin{align} \label{eq:optstrat}
		\E_{\leaf} \left[ \max_{i \in P_{\leaf}} \{X_i\} - \price(P_{\leaf}) \right] .
\end{align}
Now we design an adaptive strategy for the stochastic probing problem on random variables $\Ymax_i$ with  expected value at least  as given by Eq.~\eqref{eq:optstrat}. Consider the adaptive strategy that follows the same  decision tree $T$ (note, it pays no probing price). The expected value of such an adaptive strategy is given by
\begin{align} \label{eq:adapstrat}
		\E_{\leaf} \left[ \max_{i \in P_{\leaf} } \{ \Ymax_i \} \right] .
\end{align}

The following claim finishes the proof of this lemma.
\begin{claim}
Eq.~\eqref{eq:optstrat} $\leq $ Eq.~\eqref{eq:adapstrat}.
\end{claim}
\proof{
 Let $A_{\leaf}(i)$ and $\one_{i\in P_{\leaf}}$  denote  indicator variables that element $i$ is selected and  probed on a root-leaf path $P_{\leaf}$ of the optimal strategy, respectively. Note that these indicator variables are correlated. The expected utility of the optimal strategy equals
\begin{align*}
\E_{\leaf} \left[ \max_{i \in P_{\leaf} } \{X_i\} - \price(P_{\leaf}) \right] &= \E_{\leaf} \left[\sum_i \left( A_{\leaf}(i) X_i - \one_{i\in P_{\leaf}} \price_i \right) \right] \\
&= \E_{\leaf} \left[\sum_i \left( A_{\leaf}(i) X_i - \one_{i\in P_{\leaf}} \E_i[ (X_i - \taumax_i)^+] \right) \right] \\
&= \E_{\leaf} \left[\sum_i \left( A_{\leaf}(i) X_i - \one_{i\in P_{\leaf}}  (X_i - \taumax_i)^+ \right) \right] & \\
\intertext{since $X_i$ is independent of $\one_{i\in P_{\leaf}}$. Now using $A_{\leaf}(i) \leq \one_{i\in P_{\leaf}}$,}
&\leq \E_{\leaf} \left[\sum_i \left( A_{\leaf}(i) X_i - A_{\leaf}(i)  (X_i - \taumax_i)^+ \right) \right] & \\
&= \E_{\leaf} \left[ \sum_i A_{\leaf}(i) \Ymax_i \right]  \\
& \leq  \E_{\leaf} \left[ \max_{i \in P_{\leaf} } \{ \Ymax_i \} \right], & 
\end{align*} 
where the last inequality uses $\sum_{i\in P_{\leaf}} A_{\leaf}(i) \leq 1$.
}

\end{proof}
\else
We present the proof of this lemma in the full version. 
\fi
The following Lemma~\ref{lem:constrainedBoundAdapGap} shows  that we can further simplify the stochastic probing problem in the \FreeInfo world by  focusing only on finding the best \emph{non-adaptive} strategy for this problem, i.e. the problem of finding $\argmax_{\probed \in \J} \{\max_{i\in \probed} \{ Y_i\} \}$. This is because the \emph{adaptivity gap}---ratio of the expected values of the optimal adaptive and optimal non-adaptive strategies---for the stochastic probing problem is small.

\begin{lemma}\label{lem:constrainedBoundAdapGap} 
The adaptivity gap for the stochastic probing problem  is at most $3$.
\end{lemma}

We prove Lemma~\ref{lem:constrainedBoundAdapGap} in \ifFULL Section~\ref{sec:constrainedBoundAdapGap}. \else the full version by generalizing a similar result for Bernoulli random variables of Gupta et al.~\cite{GNS-SODA17} to functions that are given by weighted matroid rank function. \fi It  tells us about the  existence of a feasible set $S \in \J$ such that $\E[\max_{i \in S} \{ \Ymax_i \}]$ is at least $1/3$ times the optimal adaptive strategy for the stochastic probing problem. Suppose we have an oracle to (approximately) find this feasible set $S$ for probing constraints $\J$.

\begin{assumption}\label{assum:oracleNonAdaptive} Suppose there exists an oracle that finds $S \in \J$ that $\beta$ approximately maximizes  the non-adaptive stochastic probing solution.
\end{assumption}
\noindent The above assumption is justifiable as it is a constrained submodular maximization problem that we know how to approximately solve for some many constraint families $\J$, e.g., an $\ell$-system.
\begin{lemma}[Greedy Algorithm~\cite{FNW2-Journal78,CCPV-SICOMP11}] \label{lem:LSystemGreedy} The greedy algorithm has an $(\ell+1)$-approximation for monotone submodular maximization over an $\ell$-system.
\end{lemma}

Finally, we  need to show that given $S$ there exists an efficient adaptive strategy in the \POI world with expected utility $\E[\max_{i \in S} \{ \Ymax_i \}]$. But this is exactly the Pandora's box problem for which we know that Weitzman's index-based  policy is optimal with expected utility $\E[\max_{i \in S} \{ \Ymax_i \}]$. The above discussion can be summarized in the following theorem.

\begin{theorem}\label{thm:WeitzToNonAdap}
Given   a $\beta$-approximation oracle  for  monotone submodular maximization over downward-closed constraints $\J$,  there exists  a $3  \beta$-approximation algorithm for constrained utility-maximization.%  in the \POI world.
\end{theorem}

Combining Lemma~\ref{lem:LSystemGreedy} and Theorem~\ref{thm:WeitzToNonAdap}, we get Theorem~\ref{thm:constrainedUM} as a corollary.

%------------------------------------------------------------------------------------

\ifFULL
\subsection{Bounding the Adaptivity Gap} \label{sec:constrainedBoundAdapGap}
In this section we prove Lemma~\ref{lem:constrainedBoundAdapGap} by generalizing a similar result for Bernoulli variables of Gupta et al.~\cite{GNS-SODA17} to functions that are given by weighted matroid rank function. We do this by reducing the problem for discrete random variables to a Bernoulli setting.

\begin{lemma} \label{lem:adapGapToBern}
 The adaptivity gap for the stochastic probing problem  
%with weighted matroid rank function 
for discretely distributed non-negative random variables is bounded by that for Bernoulli distributed non-negative random variables.
\end{lemma}
\begin{proof}
Let us assume that each random variable $Y_i$ takes value in a discrete set $\{v_1, v_2, \ldots, v_m \}$, where it takes value $v_j$ w.p. $p_j$ and $\sum_j p_j = 1$. (Note that $p_j$ is a function of $i$ but we don't write index $i$ for ease of exposition.) Consider the optimal adaptive strategy decision tree $\T$ for the stochastic probing problem; here each node has at most $m$ children. We modify $\T$ to obtain a binary decision tree $\T'$, which shows one can transform the instance to an instance with Bernoulli random variables and the same adaptivity gap. The idea is to replace every node $i$ in $\T$ with $m$ binary decision variables in $\T'$, where variable $j$ is active w.p. $\frac{p_j}{1- \sum_{k<j}p_k}$. If active, variable $j$ leads to the subtree corresponding to the case when $Y_i$ take value $v_j$ in $\T$. The two trees are equivalent because the probability that variable $Y_i$ takes value $v_j$ in $\T'$ is exactly 
\[ \prod_{j'<j} \left(1- \frac{p_{j'}}{1- \sum_{k<j'}p_k} \right) \cdot \frac{p_j}{1- \sum_{k<j}p_k} \quad = \quad \prod_{j'<j} \left( \frac{1- \sum_{k\leq j'}p_k}{1- \sum_{k<j'}p_k} \right) \cdot \frac{p_j}{1- \sum_{k<j}p_k} \quad = \quad p_j .
\]
\end{proof}

To finish the proof of Lemma~\ref{lem:constrainedBoundAdapGap}, we combine Lemma~\ref{lem:adapGapToBern} with the following result of Gupta et al.
\begin{lemma}[\cite{GNS-SODA17}]\label{lem:GNSadapGap} 
The adaptivity gap for the stochastic probing problem for  Bernoulli random variables over any given downward-closed constraints is at most $3$.
\end{lemma}

\IGNORE{
\begin{lemma} The Bernoulli adaptivity gap for stochastic probing problem over any downward-closed set system for weighted rank function of an $\ell$-system is at most $\ell+2$.
\end{lemma}
%\begin{proof}
The proof is similar to the one presented in~\cite{GNS-SODA17} for non-monotone submodular functions. We defer it to the full version of the paper.
%\end{proof}
}
\else
\fi
%------------------------------------------------------------

\subsection{An Application to the Set-Probing Utility-Maximization Problem}\label{section:setProbing}
In this section we see an application of the constrained utility-maximization framework to the \emph{set-probing utility-maximization} problem  defined in Section~\ref{sec:introConstrainedUtilMax}. This problem is a generalization of Pandora's box  where we pay a price to simultaneously find values of a set of random variables. 
\ifFULL We restate Theorem~\ref{thm:setprobing} for convenience.
\setProbing*
\noindent  The remaining section discusses the approximation algorithm. See the hardness proof in  Section~\ref{lem:hardnessNonDisj}. 
\else 
We next discuss the approximation algorithm of Theorem~\ref{thm:setprobing} (hardness result in the full version is by reducing to $\ell$-set packing). 
\fi 

We first observe that WLOG one can assume  that the given sets $\S = \{S_1, \ldots, S_m\}$ are downward-closed, i.e., if $S \in \S $ then any subset $T\subseteq S$ is also in $\S$. This is because a simple way to ensure downward-closedness is by adding every subset of $S_j$  for the same price $\price_j$ into $\S$. Intuitively, this is equivalent to paying for the original set but choosing not to see the outcome of some of the random variables in it.

To construct our algorithm, we imagine solving a constrained utility-maximization problem. The
random variables of this problem are indexed by sets $S\in \S$: variable  $X_{S}$ has value $\max_{i\in S} \{X_i\}$ and has price $\price_S$. The problem is to adaptively probe some elements such that the sets corresponding to them are pairwise disjoint (a downward-closed constraint $\J$), while the goal is to maximize the utility that is given by max element value minus the total probing prices. Intuitively, the reason we need disjointness is to ensure independence between sets in our analysis as  disjoint sets of random variables take values independently. We  make the following simple observation.

\begin{observation}
The optimal adaptive policy for this  constrained utility-maximization problem with disjointness constraints  is the same as the  unconstrained set-probing utility-maximization problem.
\end{observation}

Given the above observation and noting that  disjointness constraints are downward-closed, we want to use Theorem~\ref{thm:WeitzToNonAdap} to reduce our problem into a non-adaptive optimization problem. Although it appears that this is not possible because Theorem~\ref{thm:WeitzToNonAdap}  is only for independent variables, and variables corresponding to non-disjoint sets are not independent. Fortunately, the proof of Theorem~\ref{thm:WeitzToNonAdap} only uses independence of random variables along any root-leaf path of the decision tree. Since our probing constraints ensure that the probed sets are  disjoint, we  get  variables $X_S$ along any root-leaf path  to be independent, thereby allowing us to use Theorem~\ref{thm:WeitzToNonAdap}. The final  part in the proof of Theorem~\ref{thm:setprobing}  is an approximation  algorithm for this  non-adaptive constrained utility-maximization problem. 

\begin{lemma}\label{lemma:nonAdapSetProbing}
There exists an efficient $(\ell+1)$-approximation  algorithm for the non-adaptive problem of finding a family $\S' \subseteq \S$ of disjoint sets   to maximize  	$E[\max_{S\in \S'} \{\Ymax_S \}]$.
\end{lemma}

\ifFULL
\begin{proof}
 Observe that the  function $g(\S')=E[\max_{S\in \S'} \{\Ymax_S\}]$ is submodular. Also, the disjointness constraints can be viewed as an $\ell$-system constraints since each set $S$ has size at most $\ell$. Thus we can view the non-adaptive problem as maximizing a submodular function over an $\ell$-system, where we know by Lemma~\ref{lem:LSystemGreedy} that the greedy algorithm has an $(\ell+1)$-approximation.

Moreover, to implement the greedy algorithm efficiently, we note that although $\S$ may contain an exponential number of elements, the initial  set system $\S$ was polynomial sized (before we made $\S$ downward-closed). For sets $A,B$ available at the same price, where $A\subseteq B$, it is obvious that the greedy algorithm will always choose $B$ before $A$. Hence, at every step our greedy algorithm only needs to consider the original sets, which are only polynomial in number, and select the set with the best marginal value. Since, this can be  done in polynomial time, this completes the proof of Theorem~\ref{thm:setprobing}.
\end{proof}
\else
\fi

\ifFULL \else \noindent See proof of Lemma~\ref{lemma:nonAdapSetProbing} in the full version.\fi

%%%%%%%%%%%%%%%%%%%%%%%%%%%%%%%%%%%%%%%
%\ifFULL

\medskip
\noindent
%{\bf Acknowledgments}.
\paragraph{Acknowledgments}
We thank Anupam Gupta, Guru Prashanth,  Bobby Kleinberg, and Akanksha Bansal for helpful comments on improving the presentation of this paper.

%\else
%\fi
%%%%%%%%%%%%%%%%%%%%%%%%%%%%%%

\bibliographystyle{alpha}
\bibliography{bib}

\newcommand{\etalchar}[1]{$^{#1}$}
\begin{thebibliography}{BYGNR98}

\bibitem[Ada11]{A11}
Marek Adamczyk.
\newblock Improved analysis of the greedy algorithm for stochastic matching.
\newblock {\em Inf. Process. Lett.}, 111(15):731--737, 2011.

\bibitem[AGM15]{AGM15}
Marek Adamczyk, Fabrizio Grandoni, and Joydeep Mukherjee.
\newblock Improved approximation algorithms for stochastic matching.
\newblock In {\em Algorithms-ESA 2015}, pages 1--12. Springer, 2015.

\bibitem[AH15]{AbbasH-Book15}
Ali~E Abbas and Ronald~A Howard.
\newblock {\em Foundations of decision analysis}.
\newblock Pearson Higher Ed, 2015.

\bibitem[ANS08]{ANS-WINE08}
Arash Asadpour, Hamid Nazerzadeh, and Amin Saberi.
\newblock Stochastic submodular maximization.
\newblock In {\em International Workshop on Internet and Network Economics},
  pages 477--489. Springer, 2008.
\newblock Full version appears as~\cite{AN16}.

\bibitem[ASW16]{ASW14}
Marek Adamczyk, Maxim Sviridenko, and Justin Ward.
\newblock Submodular stochastic probing on matroids.
\newblock {\em Mathematics of Operations Research}, 41(3):1022--1038, 2016.

\bibitem[BCN{\etalchar{+}}15]{BCNSX15}
Alok Baveja, Amit Chavan, Andrei Nikiforov, Aravind Srinivasan, and Pan Xu.
\newblock Improved bounds in stochastic matching and optimization.
\newblock In {\em APPROX}, pages 124--134, 2015.

\bibitem[BG96]{BG-AI96}
Ann Becker and Dan Geiger.
\newblock Optimization of pearl's method of conditioning and greedy-like
  approximation algorithms for the vertex feedback set problem.
\newblock {\em Artificial Intelligence}, 83(1):167--188, 1996.

\bibitem[BGK11]{BGK-SODA11}
Anand Bhalgat, Ashish Goel, and Sanjeev Khanna.
\newblock Improved approximation results for stochastic knapsack problems.
\newblock In {\em SODA}, pages 1647--1665, 2011.

\bibitem[BGL{\etalchar{+}}12]{BGLMNR-Algorithmica12}
Nikhil Bansal, Anupam Gupta, Jian Li, Juli{\'a}n Mestre, Viswanath Nagarajan,
  and Atri Rudra.
\newblock {When LP Is the Cure for Your Matching Woes: Improved Bounds for
  Stochastic Matchings}.
\newblock {\em Algorithmica}, 63(4):733--762, 2012.

\bibitem[BN14]{BN-IPCO14}
Nikhil Bansal and Viswanath Nagarajan.
\newblock On the adaptivity gap of stochastic orienteering.
\newblock In {\em IPCO}, pages 114--125, 2014.

\bibitem[BYE81]{BE-JAL81}
Reuven Bar-Yehuda and Shimon Even.
\newblock A linear-time approximation algorithm for the weighted vertex cover
  problem.
\newblock {\em Journal of Algorithms}, 2(2):198--203, 1981.

\bibitem[BYGNR98]{BGNR-SICOMP98}
Reuven Bar-Yehuda, Dan Geiger, Joseph Naor, and Ron~M. Roth.
\newblock Approximation algorithms for the feedback vertex set problem with
  applications to constraint satisfaction and bayesian inference.
\newblock {\em SIAM journal on computing}, 27(4):942--959, 1998.

\bibitem[CCPV11]{CCPV-SICOMP11}
Gruia Calinescu, Chandra Chekuri, Martin P{\'a}l, and Jan Vondr{\'a}k.
\newblock Maximizing a monotone submodular function subject to a matroid
  constraint.
\newblock {\em SIAM Journal on Computing}, 40(6):1740--1766, 2011.

\bibitem[CFG{\etalchar{+}}02]{CFGKRS-Journal02}
Moses Charikar, Ronald Fagin, Venkatesan Guruswami, Jon~M. Kleinberg, Prabhakar
  Raghavan, and Amit Sahai.
\newblock Query strategies for priced information.
\newblock {\em J. Comput. Syst. Sci.}, 64(4):785--819, 2002.

\bibitem[CHKK15]{CHHKK-COLT15}
Yuxin Chen, S~Hamed Hassani, Amin Karbasi, and Andreas Krause.
\newblock Sequential information maximization: When is greedy near-optimal?
\newblock In {\em Conference on Learning Theory}, pages 338--363, 2015.

\bibitem[CIK{\etalchar{+}}09]{CIKMR09}
Ning Chen, Nicole Immorlica, Anna~R. Karlin, Mohammad Mahdian, and Atri Rudra.
\newblock {Approximating Matches Made in Heaven}.
\newblock In {\em ICALP (1)}, pages 266--278, 2009.

\bibitem[CJK{\etalchar{+}}15]{CJKBSK-AAAI15}
Yuxin Chen, Shervin Javdani, Amin Karbasi, J~Andrew Bagnell, Siddhartha~S
  Srinivasa, and Andreas Krause.
\newblock Submodular surrogates for value of information.
\newblock In {\em AAAI}, pages 3511--3518, 2015.

\bibitem[DGV04]{DGV-FOCS04}
Brian~C. Dean, Michel~X. Goemans, and Jan Vondr{\'a}k.
\newblock Approximating the stochastic knapsack problem: The benefit of
  adaptivity.
\newblock In {\em Foundations of Computer Science, 2004. Proceedings. 45th
  Annual IEEE Symposium on}, pages 208--217. IEEE, 2004.

\bibitem[DGV05]{DGV05}
Brian~C. Dean, Michel~X. Goemans, and Jan Vondr{\'a}k.
\newblock Adaptivity and approximation for stochastic packing problems.
\newblock In {\em SODA}, pages 395--404, 2005.

\bibitem[DTW03]{DTW-SIDMA03}
Ioana Dumitriu, Prasad Tetali, and Peter Winkler.
\newblock On playing golf with two balls.
\newblock {\em SIAM Journal on Discrete Mathematics}, 16(4):604--615, 2003.

\bibitem[FNW78]{FNW2-Journal78}
Marshall~L. Fisher, George~L. Nemhauser, and Laurence~A. Wolsey.
\newblock An analysis of approximations for maximizing submodular set
  functions—ii.
\newblock In {\em Polyhedral combinatorics}, pages 73--87. Springer, 1978.

\bibitem[GGM10]{GGM-TALG10}
Ashish Goel, Sudipto Guha, and Kamesh Munagala.
\newblock How to probe for an extreme value.
\newblock {\em ACM Transactions on Algorithms (TALG)}, 7(1):12, 2010.

\bibitem[Git74]{GJ-Journal74}
John Gittins.
\newblock A dynamic allocation index for the sequential design of experiments.
\newblock {\em Progress in statistics}, pages 241--266, 1974.

\bibitem[GK01]{GK-FOCS01}
Anupam Gupta and Amit Kumar.
\newblock Sorting and selection with structured costs.
\newblock In {\em Foundations of Computer Science, 2001. Proceedings. 42nd IEEE
  Symposium on}, pages 416--425. IEEE, 2001.

\bibitem[GKMR11]{GKMR-FOCS11}
Anupam Gupta, Ravishankar Krishnaswamy, Marco Molinaro, and R.~Ravi.
\newblock Approximation algorithms for correlated knapsacks and non-martingale
  bandits.
\newblock In {\em FOCS}, pages 827--836, 2011.

\bibitem[GKNR12]{GKNR-SODA12}
Anupam Gupta, Ravishankar Krishnaswamy, Viswanath Nagarajan, and R.~Ravi.
\newblock Approximation algorithms for stochastic orienteering.
\newblock In {\em SODA}, 2012.

\bibitem[GM07]{GM-STOC07}
Sudipto Guha and Kamesh Munagala.
\newblock Approximation algorithms for budgeted learning problems.
\newblock In {\em STOC}, pages 104--113. 2007.
\newblock Full version as: \emph{Approximation Algorithms for Bayesian
  Multi-Armed Bandit Problems}, \url{http://arxiv.org/abs/1306.3525}.

\bibitem[GM09]{GuhaM09}
Sudipto Guha and Kamesh Munagala.
\newblock Multi-armed bandits with metric switching costs.
\newblock In {\em ICALP}, pages 496--507, 2009.

\bibitem[GM12]{GM-SODA07TALG12}
Sudipto Guha and Kamesh Munagala.
\newblock Adaptive uncertainty resolution in bayesian combinatorial
  optimization problems.
\newblock {\em ACM Transactions on Algorithms (TALG)}, 8(1):1, 2012.

\bibitem[GMS07]{GMS-Transactions07}
Sudipto Guha, Kamesh Munagala, and Saswati Sarkar.
\newblock Information acquisition and exploitation in multichannel wireless
  systems.
\newblock In {\em IEEE Transactions on Information Theory}. Citeseer, 2007.

\bibitem[GN13]{GN-IPCO13}
Anupam Gupta and Viswanath Nagarajan.
\newblock A stochastic probing problem with applications.
\newblock In {\em IPCO}, pages 205--216, 2013.

\bibitem[GNS16]{GNS-SODA16}
Anupam Gupta, Viswanath Nagarajan, and Sahil Singla.
\newblock Algorithms and adaptivity gaps for stochastic probing.
\newblock In {\em Proceedings of the Twenty-Seventh Annual ACM-SIAM Symposium
  on Discrete Algorithms}, pages 1731--1747. SIAM, 2016.

\bibitem[GNS17]{GNS-SODA17}
Anupam Gupta, Viswanath Nagarajan, and Sahil Singla.
\newblock {Adaptivity Gaps for Stochastic Probing: Submodular and XOS
  Functions}.
\newblock In {\em Proceedings of the Twenty-Eighth Annual ACM-SIAM Symposium on
  Discrete Algorithms}, pages 1688--1702. SIAM, 2017.

\bibitem[GW95]{GW-SICOMP95}
Michel~X Goemans and David~P Williamson.
\newblock A general approximation technique for constrained forest problems.
\newblock {\em SIAM Journal on Computing}, 24(2):296--317, 1995.

\bibitem[Has96]{Hastad-FOCS96}
Johan Hastad.
\newblock Clique is hard to approximate within $n^{1-\epsilon}$.
\newblock In {\em Foundations of Computer Science, 1996. Proceedings., 37th
  Annual Symposium on}, pages 627--636. IEEE, 1996.

\bibitem[HKT00]{HMJJ-DAM00}
Magn{\'u}s~M. Halld{\'o}rsson, Jan Kratochv{\i}l, and Jan~Arne Telle.
\newblock Independent sets with domination constraints.
\newblock {\em Discrete Applied Mathematics}, 99(1):39--54, 2000.

\bibitem[HSS06]{HSS-CC06}
Elad Hazan, Shmuel Safra, and Oded Schwartz.
\newblock On the complexity of approximating k-set packing.
\newblock {\em Computational Complexity}, 15(1):20--39, 2006.

\bibitem[Jen76]{Jenkyns-76}
Thomas~A. Jenkyns.
\newblock The efficacy of the “greedy” algorithm.
\newblock In {\em Proc. of 7th South Eastern Conference on Combinatorics, Graph
  Theory and Computing}, pages 341--350, 1976.

\bibitem[JMM{\etalchar{+}}03]{JMMSV-JACM03}
Kamal Jain, Mohammad Mahdian, Evangelos Markakis, Amin Saberi, and Vijay~V
  Vazirani.
\newblock Greedy facility location algorithms analyzed using dual fitting with
  factor-revealing lp.
\newblock {\em Journal of the ACM (JACM)}, 50(6):795--824, 2003.

\bibitem[KH78]{KH-DM78}
Bernhard Korte and Dirk Hausmann.
\newblock An analysis of the greedy heuristic for independence systems.
\newblock {\em Annals of Discrete Mathematics}, 2:65--74, 1978.

\bibitem[KK03]{kannan2003selection}
Sampath Kannan and Sanjeev Khanna.
\newblock Selection with monotone comparison costs.
\newblock In {\em Proceedings of the fourteenth annual ACM-SIAM symposium on
  Discrete algorithms}, pages 10--17. Society for Industrial and Applied
  Mathematics, 2003.

\bibitem[KWW16]{KWW-EC16}
Robert Kleinberg, Bo~Waggoner, and Glen Weyl.
\newblock {Descending Price Optimally Coordinates Search}.
\newblock {\em arXiv preprint arXiv:1603.07682}, 2016.

\bibitem[LY13]{LiYuan-STOC13}
Jian Li and Wen Yuan.
\newblock Stochastic combinatorial optimization via poisson approximation.
\newblock In {\em Symposium on Theory of Computing Conference, STOC'13, Palo
  Alto, CA, USA, June 1-4, 2013}, pages 971--980, 2013.

\bibitem[Ma14]{Ma-SODA14}
Will Ma.
\newblock Improvements and generalizations of stochastic knapsack and
  multi-armed bandit approximation algorithms: Extended abstract.
\newblock In {\em SODA}, pages 1154--1163, 2014.

\bibitem[Wei79]{Weitzman-Econ79}
Martin~L. Weitzman.
\newblock Optimal search for the best alternative.
\newblock {\em Econometrica: Journal of the Econometric Society}, pages
  641--654, 1979.

\bibitem[WS11]{WS-Book11}
David~P. Williamson and David~B. Shmoys.
\newblock {\em The design of approximation algorithms}.
\newblock Cambridge university press, 2011.

\end{thebibliography}

\appendix

%%%%%%%%%%%%%%%%%%%%%%%%%%%%%%%%%%%%%%%
\ifFULL

\section{Illustrative Examples} \label{sec:examples}

\subsection{Why the na\"{\i}ve greedy algorithm fails for Pandora's box} \label{sec:egNaiveGreedy}

Suppose $curr$ denotes the maximum value in the currently opened set of boxes.
The na\"{\i}ve greedy  algorithm selects in any step the unopened box $j$ corresponding to the maximum marginal value, i.e. $\argmax\{ \E[(X_j - curr)^+] - \price_j\} $, and opens it if its marginal value is non-negative. The algorithm stops probing when every unopened box has a negative marginal value.  We give an example where this algorithm can be made arbitrarily worse as compared to the optimal algorithm.

Consider  $n-1$ iid boxes, each taking value $1/p^2$ w.p. $p$ and $0$ otherwise, where $p<1$. The probing price of each of these boxes is $1$. Also, there is a box which takes value $1/p^2$ w.p. $1$ but has a probing price of $1/p^2- 1/p +1$. Note that in the beginning, the marginal value of every box is $1/p-1$. Now the optimal strategy is to probe the boxes with price $1$, until we see a box with value $1/p^2$. For large enough $n$, the expected utility of this strategy is $\approx 1/p^2 - 1/p$ because in expectation the algorithm stops after roughly $1/p$ probes. However, the na\"{\i}ve greedy algorithm will open the box with  price $1/p^2- 1/p-1$ and then stop probing. Thus its expected utility will be $1/p-1$. By choosing small enough $p$, the ratio between $1/p^2 - 1/p$ and $1/p-1$ can be made arbitrarily large.

%--------------------------------------------------------
\subsection{The Pandora's box problem has no constant approximation non-adaptive solution} \label{sec:adapGapHardness}
Consider an example where each element independently takes value $1$ w.p. $p$ ($\ll 1$) and value $0$, otherwise. Suppose the price of probing any element is $1-\epsilon$, for some small $\epsilon>0$. The optimal adaptive strategy is to continue probing till we see an element with value $1$. Assuming $n$ to be large, it is easy to see that this strategy has expected value $\approx \epsilon/p$. On the other hand, a non-adaptive strategy has to decide in the beginning which all elements $S$ to probe, and then probe them irrespective of the consequence. Since all items are identical, the only decision it has to make is how many items to probe. One can verify that no such non-adaptive strategy can get value more than $O(\epsilon)$. By choosing $p$ to be small enough, we can make the gap arbitrarily large.

%--------------------------------------------------------
\subsection{Hardness for general submodular functions} \label{sec:submodHardness}
To prove that one cannot obtain good approximation results for any monotone submodular  functions $f$, we show that even when all variables are deterministic, the computational problem of selecting the best set $\I \subseteq V$ to maximize $f(\I) - \price(\I)$ is  $\tilde{\Omega}(\sqrt{n})$ hard assuming $P\neq NP$, where $n=|V|$. The idea is to reduce from set packing. Let $\S=\{S_1, S_2, \ldots, S_m\}$ denote the sets of a set packing instance. For $S\subseteq \S$, let $f(S)$ denote the covering function. Let price of probing $S_i$ be $\price_i = |S_i|-1$. Clearly, it doesn't make sense to probe sets that are not disjoint as the marginal utility will be non-positive. The optimal solution therefore equals the maximum number of disjoint sets. But no polynomial time algorithm can be $O({n^{1/2-\epsilon}})$-approximation, for any $\epsilon>0$, unless $P=NP$~\cite{Hastad-FOCS96,HMJJ-DAM00}.

%--------------------------------------------------------
\subsection{Hardness for the set-probing problem} \label{lem:hardnessNonDisj}
To prove that no polynomial time algorithm for the set-probing problem can be  $o(\ell/\log \ell)$-approximation, unless $P=NP$, we reduce  $\ell$-set packing problem into an instance of the set-probing problem.
Given an $\ell$-set packing instance with sets $S_1, S_2, \ldots, S_m$, each of size $\ell$, we create the following set-probing problem. For every element $i \in \bigcup_j S_j$, w.p.  $\frac{1}{n^3}$ variable $X_i$ takes value $1$, and is $0$, otherwise. Also, let each set $S_j$ have price $\frac{(\ell-0.5)}{n^3}$. 

Since probability of two elements taking value $1$ is really small $O(1/n^4)$, we  see that it only makes sense to probe sets where none of the elements have been already probed: even if a single element is probed before, expected value from probing the set is at most $\frac{(\ell-1)}{n^3} - \frac{(\ell-0.5)}{n^3} = \frac{-0.5}{n^3} < 0$.
Hence, $E[Opt] = (\text{Max \# disjoint sets})\cdot \frac{0.5}{n^3}$.
But this is exactly the $\ell$-set packing problem and we know that unless $P=NP$, no poly time algorithm can be $o(\ell/\log \ell)$-approximation~\cite{HSS-CC06}.

\else
\fi

\end{document}